\documentclass[12pt]{article}
\usepackage{amsmath}
\usepackage{graphicx,psfrag,epsf}
\usepackage{enumerate}
\usepackage{psfrag,epsf}
\usepackage{url} 
\usepackage{amsmath,amssymb,amsfonts,amsthm,mathtools}
\usepackage{bm,bbm}
\usepackage{color}
\usepackage{subfigure}
\usepackage{algorithm,algpseudocode}
\usepackage{scalefnt}
\usepackage{multirow}

\usepackage{natbib}
\usepackage{dsfont}

\usepackage[left=1in,top=1.2in,right=0.7in]{geometry}

\usepackage{letltxmacro}
\makeatletter
\let\oldr@@t\r@@t
\def\r@@t#1#2{%
	\setbox0=\hbox{$\oldr@@t#1{#2\,}$}\dimen0=\ht0
	\advance\dimen0-0.2\ht0
	\setbox2=\hbox{\vrule height\ht0 depth -\dimen0}%
	{\box0\lower0.4pt\box2}}
\LetLtxMacro{\oldsqrt}{\sqrt}
\renewcommand*{\sqrt}[2][\ ]{\oldsqrt[#1]{#2}}
\makeatother
\theoremstyle{definition}

\newtheorem{theorem}{Theorem}[section]

\newtheorem{corollary}[theorem]{Corollary}

\newtheorem{proposition}[theorem]{Proposition}
\newtheorem{remark}[theorem]{Remark}
\numberwithin{equation}{section} 

\makeatletter
\def\@seccntformat#1{\@ifundefined{#1@cntformat}%
	{\csname the#1\endcsname\quad}
	{\csname #1@cntformat\endcsname}
}
\makeatother

\newif\ifShowComments
\ShowCommentstrue
\def\strutdepth{\dp\strutbox}
\def\druk#1{\strut\vadjust{\kern-\strutdepth
        {\vtop to \strutdepth{%
                \baselineskip\strutdepth\vss
                        \llap{\hbox{#1}\quad}\null}}}}

\title{\bf
Symmetric generalized Heckman models
}
\author{
\large
{Helton Saulo}, {Roberto Vila} and {Shayane S. Cordeiro}\\
{\small Department of Statistics, Universidade de Bras\'{i}lia, Bras\'{i}lia, Brazil}\\
}

\begin{document}
\maketitle

\begin{abstract}
The sample selection bias problem arises when a variable of interest is correlated with a latent variable, and involves situations in which the response variable had part of its observations censored. \cite{heckman76} proposed a sample selection model based on the bivariate normal distribution that fits both the variable of interest and the latent variable. Recently, this assumption of normality has been relaxed by more flexible models such as the Student-$t$ distribution \citep{Genton2012,lachosetal21}. The aim of this work is to propose generalized Heckman sample selection models based on symmetric distributions \citep{fkn:90}. This is a new class of sample selection models, in which variables are added to the dispersion and correlation parameters. A Monte Carlo simulation study is performed to assess the behavior of the parameter estimation method. Two real data sets are analyzed to illustrate the proposed approach.

\end{abstract}


\section{Introduction}

It is common in the areas of economics, statistics, sociology, among others, that in the sampling process there is a relationship between a variable of interest and a latent variable, in which the former is observable only in a subset of the population under study. This problem is called sample selection bias and was studied by \cite{heckman76}. The author proposed a sample selection model by joint modeling the variable of interest and the latent variable. The classical Heckman sample selection (classical Heckman-normal model) model received several criticisms, due to the need to assume bivariate normality and the difficulty in estimating the parameters using the maximum likelihood (ML) method, which led to the introduction of an alternative estimation method known as the two-step method; see \cite{heckman79}. Some studies on Heckman models have been done by \cite{Nelson1984}, \cite{Paarsch1984}, \cite{Manning1987}, \cite{Stolzenberg1990} and \cite{Yu1996}. These works suggested that the Heckman sample selection model can reduce or eliminate selection bias when the assumptions hold, but deviation from normality assumption may distort the results.

The normality assumption of the classical Heckman-normal model \citep{heckman76} has been relaxed by more flexible models such as the Student-$t$ distribution \citep{Genton2012,Ding2014,lachosetal21}, the skew-normal distribution \citep{Ogundimu_2016} and the Birnbaum-Saunders distribution \citep{bastosbarretosouza:21}. Moreover, the classical Heckman-normal model assumes that the dispersion and correlation (sample selection bias parameter) are constant, which may not be adequate. In this context, the present work aims to propose generalized Heckman sample selection models based on symmetric distributions \citep{fkn:90}. In the proposed model, covariates are added to the dispersion and correlation parameters, then we have covariates explaining possible heteroscedasticity and sample selection bias, respectively. Our proposed methodology can be seen as a generalization of the generalized Heckman-normal model with varying sample selection bias and dispersion parameters proposed by \cite{Bastos2021} and the Heckman-Student-$t$ model proposed by \cite{Genton2012}. We demonstrate that the proposed symmetric generalized Heckman model outperforms the generalized Heckman-normal and Heckman-Student-$t$ models in terms of model fitting, making it a practical and useful model for modelling data with sample selection bias.

The rest of this work proceeds as follows. In Section \ref{sec:2}, we briefly describe the bivariate symmetric distributions. We then introduce the symmetric generalized Heckman models. In this section, we also describe the maximum likelihood (ML) estimation of the model parameters. In Section \ref{sec:3}, we derive the generalized Heckman-Student-$t$ model, which is a special case of the symmetric generalized Heckman models. In Section \ref{sec:4}, we carry out a Monte Carlo simulation study for evaluating the performance of the estimators. In Section \ref{sec:5}, we apply the generalized Heckman-Student-$t$ to two real data sets to demonstrate the usefulness of the proposed model, and finally in Section \ref{sec:6}, we provide some concluding remarks.

\section{Symmetric generalized Heckman models}\label{sec:2}

Let $\boldsymbol{Y}=(Y_1,Y_2)^{\top}$ be a random vector following a bivariate symmetric (BSY) distribution \citep{fkn:90} with location (mean) vector
$\boldsymbol{\mu}=(\mu_1,\mu_2)^{\top}$, covariance matrix
\begin{align}\label{covariance matrix}
\boldsymbol{\Sigma}= 
\begin{pmatrix}
  \sigma_1^2 & \rho \sigma_1 \sigma_2  \\
 \rho \sigma_1 \sigma_2 & \sigma_2^2
 \end{pmatrix},
\end{align}
and density generator $g_c$,
with $\mu_i\in\mathbb{R}$, $\sigma_i>0$, for $i=1,2$.  We use the notation $\boldsymbol{Y}\sim {\rm BSY}(\boldsymbol{\mu},\boldsymbol{\Sigma},g_c)$. Then, the
probability density function (PDF) of $\boldsymbol{Y}\sim {\rm BSY}(\boldsymbol{\mu},\boldsymbol{\Sigma},g_c)$ is given by
\begin{equation}\label{eq:pdf:sym}
f_{\boldsymbol{Y}}(\boldsymbol{y}; \boldsymbol{\mu},\boldsymbol{\Sigma},g_c) 
= 
\frac{1}{|\boldsymbol{\Sigma}|^{1/2} { Z_{g_c}}}\, g_c \left((\boldsymbol{y} - \boldsymbol{\mu})^\top \boldsymbol{\Sigma}^{-1} (\boldsymbol{y} - \boldsymbol{\mu})\right), \quad \boldsymbol{y} \in \mathbb{R}^2,
\end{equation}
{ where $|\boldsymbol{\Sigma}|=\sigma_1^2\sigma_2^2(1-\rho^2)$ and $Z_{g_c}$ is a normalization constant so that $f_{\boldsymbol{Y}}$ is a PDF, that is,
	\begin{align*}
	Z_{g_c}
	=
	\int_{\mathbb{R}^2} 
	\frac{1}{|\boldsymbol{\Sigma}|^{1/2}}\, g_c \left((\boldsymbol{y} - \boldsymbol{\mu})^\top \boldsymbol{\Sigma}^{-1} (\boldsymbol{y} - \boldsymbol{\mu})\right)\, {\rm d}\boldsymbol{y}
	=
	\pi \int_{0}^{\infty} g_c(u)\,{\rm d}u.
	\end{align*}
}\noindent
The density generator $g_c$ in \eqref{eq:pdf:sym} leads to different bivariate symmetric distributions, which may contain an extra parameter (or extra parameter vector).

\bigskip
We propose a generalization of the classical Heckman-normal model \citep{heckman76} by considering independent 
errors terms following
a BSY distribution with regression structures for the sample selection bias ($0<\rho<1$) and dispersion ($\sigma>0$) parameters:
\begin{eqnarray}\label{errors:gen}
{Y^{*}_i \choose U^{*}_i} \sim \text{BSY} 
\Biggl(  \boldsymbol{\mu} = {\mu_{1i} \choose \mu_{2i}},
 \boldsymbol{\Sigma}= \begin{pmatrix}
  \sigma_{i}^2 & \rho_{i}\sigma_{i}  \\
 \rho_{i}\sigma_{i} & 1
 \end{pmatrix}, g_c
 \Biggr), \quad i=1,\ldots,n.
\end{eqnarray}
In the above equation $\mu_{1i}, \mu_{2i}, \sigma_i$ and $\rho_i$ are
are the mean, dispersion and correlation parameters, respectively, with the following regression structure
$g_1(\mu_{1i}) = \boldsymbol{x}_i^\top \boldsymbol{\beta}  $,
$g_2(\mu_{2i}) = \boldsymbol{w}_i^\top \boldsymbol{\gamma} $,
$h_1(\sigma_{i}) = \boldsymbol{z}_i^\top \boldsymbol{\lambda}$ and
$h_2(\rho_{i}) = \boldsymbol{v}_i^\top \boldsymbol{\kappa}$, where $\boldsymbol{\beta} = (\beta_1, \ldots, \beta_k)^\top \in \mathbb{R}^{k}$, $\boldsymbol{\gamma} = (\gamma_1, \ldots, \gamma_l)^\top \in \mathbb{R}^{l}$, $\boldsymbol{\lambda} = (\lambda_1, \ldots, \lambda_p)^\top \in \mathbb{R}^{p}$ and $\boldsymbol{\kappa} = (\kappa_1, \ldots, \kappa_q)^\top \in \mathbb{R}^{q}$ are vectors of regression coefficients, ${\boldsymbol{x_i}}=(x_{i1},\ldots,x_{ik})^{\top} $, ${\boldsymbol{w_i}}=(w_{i1},\ldots,w_{il})^{\top}$, ${\boldsymbol{z_i}}=(z_{i1},\ldots,z_{ip})^{\top}$ and ${\boldsymbol{v_i}}=(v_{i1},\ldots,v_{iq})^{\top}$ are the values of $k$, $l$, $p$ and $q$ covariates,
and $k+l+p+q < n$. The links $g_1(\cdot), g_2(\cdot), h_1(\cdot)$ and $h_2(\cdot)$ are strictly monotone and twice differentiable. The link functions $g_1: \mathbb{R}\rightarrow \mathbb{R}$, $g_2: \mathbb{R} \rightarrow \mathbb{R}$, $h_1: \mathbb{R}^+ \rightarrow \mathbb{R}$ and $h_2: [-1,1] \rightarrow \mathbb{R}$  must be strictly monotone, and at least twice differentiable, with $g_1^{-1}(\cdot)$, $g_2^{-1}(\cdot)$, $h_1^{-1}(\cdot)$, and $h_2^{-1}(\cdot)$ being the inverse functions of $g_1(\cdot)$, $g_2(\cdot)$, $h_1(\cdot)$, and $h_2(\cdot)$, respectively. For $g_1(\cdot)$ and $g_2(\cdot)$ the most common choice is the identity link, whereas for $h_1(\cdot)$ and $h_2(\cdot)$ the most common choices are logarithm and arctanh (inverse hyperbolic tangent) links, respectively.

We can agglutinate the information from $U_{i}^{*}$ in the following indicator function $U_i=\mathds{1}_{\{U_{i}^{*}>0\}}$.

Let $Y_{i}=Y_{i}^{*}U_i$ be the observed outcome,
for $i=1,\ldots,n$. Only $n_1$ out of $n$ observations $Y_{i}^{*}$ for which $U_{i}^{*} > 0$ are observed. This model is known as ``Type 2 tobit model'' in the econometrics literature. 
Notice that $U_i\sim {\rm Bernoulli}(\mathbb{P}(U_{i}^*> 0))$.
By using law of total probability,
for $\boldsymbol{\theta}=(\boldsymbol{\beta}^{\top},\boldsymbol{\gamma}^{\top},\sigma,\rho)^{\top}$, the random variable $Y_i$ has distribution function
	\begin{align*}
	F_{Y_i}(y_i;\boldsymbol{\theta})
	&=
	\mathbb{P}(Y_{i}\leq y_i|U_{i}^*>0) \mathbb{P}(U_{i}^*>0)
	+
	\mathbb{P}(Y_{i}\leq y_i|U_{i}^*>0) \mathbb{P}(U_{i}^*\leq 0)
	\\[0,2cm]
	&=
	\begin{cases}
	\mathbb{P}(Y_{i}^*\leq y_i|U_{i}^*>0) \mathbb{P}(U_{i}^*>0), &  \text{if}\, y_i<0,
	\\[0,2cm]
	\mathbb{P}(Y_{i}^*\leq y_i|U_{i}^*>0) \mathbb{P}(U_{i}^*>0)+ \mathbb{P}(U_{i}^*\leq 0),     &  \text{if}\, y_i\geq 0.
	\end{cases}
	\end{align*}
	The function $F_{Y_i}$ has only one jump, at $y_i=0$, and $\mathbb{P}(Y_{i}=0)=\mathbb{P}(U_{i}^*\leq 0)$. Therefore, $Y_i$ is a random variable that is neither discrete nor absolutely continuous, but a mixture of the two types. In other words,
	\begin{align*}
	F_{Y_i}(y_i;\boldsymbol{\theta})
	=
	\mathbb{P}(U_{i}^*\leq 0)
	F_{\rm d}(y_i)
	+
	\mathbb{P}(U_{i}^*> 0)
	F_{\rm ac}(y_i),
	\end{align*}
	where $F_{\rm d}(y_i)=\mathds{1}_{[0,+\infty)}(y_i)$ and $F_{\rm ac}(y_i)=\mathbb{P}(Y_{i}^*\leq y_i|U_{i}^*>0)$. Hence, the
PDF of $Y_i$ is given by
	\begin{align}\label{eq:pdfy}
	f_{Y_i}(y_i;\boldsymbol{\theta})
	&=
	\mathbb{P}(U_{i}^*\leq 0)
	\delta_0(y_i)
	+
	\mathbb{P}(U_{i}^*> 0)
	f_{Y_{i}^*|U_{i}^*>0 }(y_i;\boldsymbol{\theta})
	\nonumber
	\\[0,2cm]
	&=
	(\mathbb{P}(U_{i}^*\leq 0))^{1-u_i}
	(\mathbb{P}(U_{i}^*> 0))^{u_i}
	(f_{Y_{i}^*|U_{i}^*>0 }(y_i;\boldsymbol{\theta}))^{u_i}
	\\[0,2cm]
	&=
	\mathbb{P}(U_i = u_i) 
	(f_{Y_{i}^*|U_{i}^*>0 }(y_i;\boldsymbol{\theta}))^{u_i},
	\quad u_i=0,1,
	\nonumber
	\end{align}
	wherein $\mathbb{P}(U_i=0)=1-\mathbb{P}(U_i=1)=\mathbb{P}(U_{i}^*\leq 0)$ for $i=1,\ldots,n$,  and $\delta_0$ is the Dirac delta function.
That is, the density of $Y_i$ is composed of a discrete component described by the
probit model 
$\mathbb{P}(U_i = u_i) 
=
	(\mathbb{P}(U_{i}^*\leq 0))^{1-u_i}
	(\mathbb{P}(U_{i}^*> 0))^{u_i} 
$, for $u_i=0,1$, and a continuous part given by the conditional PDF $f_{Y_{i}^*|U_{i}^*>0 }(y_i;\boldsymbol{\theta})$.

\subsection{Finding the conditional density of ${Y}^*_i$, given that  ${U}^*_i > 0$}
In the context of sample selection models, the interest lies in finding the PDF of ${Y}^*_i|{U}^*_i > 0$ given that $({Y}^*_i, {U}^*_i)^\top$ follows the BSY distribution \eqref{errors:gen}; see Theorem \ref{Main-Theorem}.

Before stating and proving the main result (Theorem \ref{Main-Theorem}) of this section, throughout the paper we will adopt the following notations:
\begin{align}\label{int-G}
G_i(x)
=
	{\displaystyle	
	\int_{-x}^{\infty}	}
{f_{Z_{2i}\vert Z_{1i}} \Big(w_i\, \Big\vert\, {y_i-\mu_{1i}\over\sigma_i }\Big)
}
\,  {\rm d}w_i
\end{align}
and
\begin{align}\label{int-H}
H_i(x)
=
\int_{-x}^{\infty} 
f_{\rho_i Z_{1i}+\sqrt{1-\rho_i^2}\, Z_{2i}}(u_i) \,  {\rm d}u_i,
\end{align}
where
\begin{align}\label{def-Z}
Z_{1i}=RDV_{1i} \quad \text{and} \quad Z_{2i}=R\sqrt{1-D^2}V_{2i}
\end{align}
have the joint PDF $f_{Z_{1i},Z_{2i}}$ and $f_{X}$ denotes the PDF corresponding to a random variable $X$.
Here, the random variables $V_{1i}$, $V_{2i}$, $R$, and $D$ are mutually independent and $\mathbb{P}(V_{ki} = -1) = \mathbb{P}(V_{ki} = 1) = 1/2$, $k=1,2$. 
The random variable $D$ in \eqref{def-Z} is positive  and has PDF
\begin{align*}
f_D(d)={2\over \pi\sqrt{1-d^2}}, \quad d\in(0,1).
\end{align*}
The random variable $R$ in \eqref{def-Z} is positive and is called the generator of the random vector $(Y_i^*,U_i^*)^{\top}$. 
Moreover, $R$ has PDF given by
\begin{align*}
f_R(r)={2r g_c(r^2)\over \int_{0}^{\infty}
	g_c(u)
	\, {\rm d}{u}}, \quad r>0,
\end{align*}
where $g_c$ is the density generator in \eqref{eq:pdf:sym}.

\begin{proposition}\label{proposition-pdfs}
		Let us denote $S=RD$ and  $T=R\sqrt{1-D^2}$.
	\begin{enumerate}
		\item
The PDFs of $S$ and $T$ are given by
\begin{align}\label{dens-s-t}
f_{S}(v)=f_{T}(v)=
{\int_{v}^{\infty} {4g_c(w^2)\over \sqrt{1-{v^2\over w^2}}}\, {\rm d}w\over \pi \int_{0}^{\infty}g_c(u)\, {\rm d}{u}}\, ,
\quad v>0.
\end{align}		
		\item 
The PDFs of  $Z_{1i}$ and $Z_{2i}$ are 
\begin{align*}
	f_{Z_{1i}}(z)
	=
	f_{Z_{2i}}(z)
	=	
	{\int_{z}^{\infty} {2g_c(w^2)\over \sqrt{1-{z^2\over w^2}}}\, {\rm d}w\over \pi\int_{0}^{\infty}g_c(u)\, {\rm d}{u}}\,,
	\quad -\infty<z<\infty.
\end{align*}
	\end{enumerate}
\end{proposition}
\begin{proof}
	Using the known formula for the PDF of product of two random variables $X$ and $Y$:
	\begin{align*}
	f_{XY}(u)=\int_{-\infty}^{\infty}{1\over\vert x\vert}\, f_{X,Y}\Big(x,{u\over x}\Big)\, {\rm d}x
	\end{align*}
	the proof of \eqref{dens-s-t} follows.
	
	The proof of second item follows by combining the law of total probability with  \eqref{dens-s-t}.
\end{proof}

\begin{proposition}\label{prop-int}
	The random vector $(Z_{1i},Z_{2i})^{\top}$ is jointly symmetric about $(0,0)$. That is,
	$f_{Z_{1i},Z_{2i}}(x,y)=f_{Z_{1i},Z_{2i}}(x,-y)=f_{Z_{1i},Z_{2i}}(-x,y)=f_{Z_{1i},Z_{2i}}(-x,-y)$.
\end{proposition}
\begin{proof}
	Let $S=RD$ and  $T=R\sqrt{1-D^2}$.
Since  $R$ and $D$ are mutually independent, by using change of variables (Jacobian Method), the joint density of $S$ and $T$ is 
\begin{align}\label{pdf-S-T}
	f_{S,T}(s,t)
	=
	{4g_c(s^2+t^2)\over \pi  \int_{0}^{\infty} g_c(u)\, {\rm d}{u}}, \quad s,t>0.
\end{align}

Moreover, since $V_{1i}\stackrel{d}{=}V_{2i}\sim {\rm Bernoulli}(1/2)$ are mutually independent, by law of total probability we get
\begin{multline}\label{dec-cdf}
F_{Z_{1i},Z_{2i}}(x,y)
=
{1\over 4}\, 
\big\{
F_{S,T}(x,y)
\mathds{1}_{[0,\infty)\times [0,\infty)}(x,y)
+
\mathbb{P}(S\geqslant -x,T\geqslant -y)
\mathds{1}_{(-\infty,0)\times (-\infty,0)}(x,y)
\\[0,2cm]
+
\mathbb{P}(S\leqslant x,T\geqslant -y)
\mathds{1}_{(0,\infty)\times (-\infty,0)}(x,y)
+
\mathbb{P}(S\geqslant -x,T\leqslant y)
\mathds{1}_{(-\infty,0)\times (0,\infty)}(x,y)
\big\},
\end{multline}
where $F_{X,Y}(\cdot,\cdot)$ denotes the (joint) distribution function of $(X,Y)^{\top}$.
Using the following well-known identity:
\begin{align*}
\mathbb{P}(a_1<X\leqslant b_1, a_2<Y\leqslant b_2)
=
F_{X,Y}(b_1,b_2)-F_{X,Y}(b_1,a_2)-F_{X,Y}(a_1,b_2)+F_{X,Y}(a_1,a_2),
\end{align*}
we have
\begin{align*}
&\mathbb{P}(S\geqslant -x,T\geqslant -y)
=
F_{S,T}(\infty,\infty)-F_{S,T}(\infty,-y)-F_{S,T}(-x,\infty)+F_{S,T}(-x,-y), \ x<0,y<0;
\\[0,2cm]
&\mathbb{P}(S\leqslant x,T\geqslant -y)
=F_{S,T}(x,\infty)-F_{S,T}(x,-y)-F_{S,T}(0,\infty)+F_{S,T}(0,-y), \quad x>0,y<0;
\\[0,2cm]
&\mathbb{P}(S\geqslant -x,T\leqslant y)
=F_{S,T}(\infty,y)-F_{S,T}(\infty,0)-F_{S,T}(-x,y)+F_{S,T}(-x,0), \quad x<0,y>0.
\end{align*}
By replacing the last three identities in \eqref{dec-cdf} and then by differentiating $F_{Z_{1i},Z_{2i}}(x,y)$ with respect to $x$ and $y$, we obtain
\begin{align*}
	f_{Z_{1i},Z_{2i}}(x,y)
	&=
	 {1\over 4}\, \big\{
	f_{S,T}(x,y)\mathds{1}_{(0,\infty)\times (0,\infty)}(x,y)
	+
	f_{S,T}(-x,-y)\mathds{1}_{(-\infty,0)\times (-\infty,0)}(x,y)
	\\[0,2cm]
	&\quad+
	f_{S,T}(x,-y)\mathds{1}_{(0,\infty)\times (-\infty,0)}(x,y)
	+
	f_{S,T}(-x,y)\mathds{1}_{(-\infty,0)\times (0,\infty)}(x,y)
	\big\}
	\\[0,2cm]
	&={1\over 4}\,f_{S,T}(x,y)
	\mathds{1}_{\mathbb{R}\setminus\{0\}\times \mathbb{R}\setminus\{0\}}(x,y)
	=
	{g_c(x^2+y^2)\over \pi  \int_{0}^{\infty} g_c(u)\, {\rm d}{u}}
		\mathds{1}_{\mathbb{R}\setminus\{0\}\times \mathbb{R}\setminus\{0\}}(x,y),
\end{align*}
where in the last line we used the Equation \eqref{pdf-S-T}. Note that the function $f_{Z_{1i},Z_{2i}}$ above 
is not defined on the abscise or ordinate axes
(neither at the origin), but this is not relevant because these events have a null probability measure. Therefore, we can say that
\begin{align}\label{eq-joint-pdf}
	f_{Z_{1i},Z_{2i}}(x,y)
	=
	{g_c(x^2+y^2)\over \pi  \int_{0}^{\infty} g_c(u)\, {\rm d}{u}}, \quad -\infty<x,y<\infty.
\end{align}
From \eqref{eq-joint-pdf} it is clear that the vector $(Z_{1i},Z_{2i})^{\top}$ is jointly symmetric about $(0,0)$.
\end{proof}

\begin{proposition}\label{prop-marg-sym}
	\begin{enumerate}
		\item 
	The marginal PDFs of $Z_{1i}$ and $Z_{1i}$, denoted by $f_{1i}$ and $f_{2i}$, respectively, are given by
	\begin{align*}
	&f_{{1i}}(x)
	=
	{\int_{-\infty}^{\infty}g_c(x^2+y^2) \, {\rm d}y\over \pi \int_{0}^{\infty} g_c(u)\, {\rm d}{u}}
	=
	f_{Z_{1i}}(x), 
	\quad -\infty<x<\infty,
	\\[0,2cm]
	&f_{{2i}}(y)
	=
	{\int_{-\infty}^{\infty}g_c(x^2+y^2) \, {\rm d}x\over \pi \int_{0}^{\infty} g_c(u)\, {\rm d}{u}}
	=
	f_{Z_{2i}}(y),
	\quad -\infty<y<\infty,
	\end{align*}
	where $f_{Z_{1i}}$ and $f_{Z_{1i}}$ are given in Proposition \ref{proposition-pdfs}.
	
	\item 
	The random vector $(Z_{1i},Z_{2i})^{\top}$ is marginally symmetric about $(0, 0)$. That is, $f_{{1i}}(x)=f_{{1i}}(-x)$ and  $f_{{2i}}(y)=f_{{2i}}(-y)$.
		\end{enumerate}
\end{proposition}
\begin{proof}
The proof of first item is immediate from the joint density of $(Z_{1i},Z_{2i})$, given in \eqref{eq-joint-pdf}, and by Proposition \ref{proposition-pdfs}. The proof of the second item follows from the first one.
\end{proof}

\begin{proposition}\label{Prop-dual}
	The function $G_i$ in \eqref{int-G} satisfies the following identity:
\begin{align*}
G_i(x)
=
1-G_i(-x)
=
{\displaystyle	
	\int_{-\infty}^{x}	}
{f_{Z_{2i}\vert Z_{1i}} \Big(w_i\, \Big\vert\, {y_i-\mu_{1i}\over\sigma_i }\Big)
}
\,  {\rm d}w_i.
\end{align*}
That is, $G_i$ is the conditional CDF of $Z_{2i}$, given that $Z_{1i}=(y_i-\mu_{1i})/\sigma_i$. Consequently, $G_i$ is a distribution symmetric about $0$. 
\end{proposition}
\begin{proof}
	The proof immediately follows by applying Proposition \ref{prop-int}.
\end{proof}

We now proceed to establish the main result of this section.
\begin{theorem}\label{Main-Theorem}
If $({Y}^*_i, {U}^*_i)^\top\sim {\rm BSY}(\boldsymbol{\mu},\boldsymbol{\Sigma},g_c)$ then the PDF of ${Y}^*_i|{U}^*_i > 0$ is given by
\begin{eqnarray}\label{label:sym:condi}
f_{{Y}_i^*|{U}_i^*>0} (y_i;\boldsymbol{\theta})
=
{1\over\sigma_i}\, f_{Z_{1i}} \biggl({y_i-\mu_{1i}\over\sigma_i }\biggr) \,
{G_i\Big({1\over \sqrt{1-\rho_i^2}}\,\mu_{2i}+{\rho_i\over \sqrt{1-\rho_i^2}}\,  ({y_i-\mu_{1i}\over\sigma_i })\Big)\over H_i(\mu_{2i})},
\end{eqnarray}
where $G_i$, $H_i$ and $Z_{1i}$ are given in  Proposition \ref{Prop-dual}, \eqref{int-H} and \eqref{def-Z}, respectively.

As a by-product of the proof we get that $\mathbb{P}(U_i^*>0)=H_i(\mu_{2i})$.
\end{theorem}
\begin{proof}
Since the random vector $(Y_i^*,U_i^*)^{\top}$ follows the BSY distribution \eqref{errors:gen}, this one admits the stochastic representation \citep{Adous2005}
\begin{align}\label{rep-est}
\begin{array}{llcc}
&Y_i^*=\sigma_i Z_{1i}+\mu_{1i},
\\[0,4cm]
&U_i^*=\rho_i Z_{1i}+\sqrt{1-\rho_i^2}\, Z_{2i}+\mu_{2i},
\end{array}
\end{align}
where $Z_{1i}$ and $Z_{2i}$ are as in \eqref{def-Z}.
If $Y_i^*=y_i$
then $Z_{1i}=(y_i-\mu_{1i})/\sigma_i$. So, the conditional distribution of $U_i^*$,  given that $Y_i^*=y_i$ is the
same as the distribution of
\begin{align*}
\rho_i\,  \biggl({y_i-\mu_{1i}\over\sigma_i }\biggr)+\sqrt{1-\rho_i^2} \, Z_{2i}+\mu_{2i}\ \Big\vert \ Y_i^*=y_i.
\end{align*}
Consequently, the PDF of $U_i^*$ given that $Y_i^*=y_i$ is given by
\begin{align}\label{id-1}
	f_{U_i^*\vert Y_i^*}(u_i\vert y_i)
	=
\dfrac{f_{Z_{1i},\, Z_{2i}} \Big({y_i-\mu_{1i}\over\sigma_i }, {1\over \sqrt{1-\rho_i^2}}\,(u_i- \mu_{2i}) -{\rho_i\over \sqrt{1-\rho_i^2}}\,  ({y_i-\mu_{1i}\over\sigma_i })\Big)
}
{\sqrt{1-\rho_i^2} \, f_{Z_{1i}} ({y_i-\mu_{1i}\over\sigma_i })}.
\end{align}
Further,
\begin{align}\label{id-2}
	f_{Y_i^*}(y_i)
=
{1\over\sigma_i}\, f_{Z_{1i}} \biggl({y_i-\mu_{1i}\over\sigma_i }\biggr).
\end{align}
By using the identity
\begin{equation*}
	f_{{Y}_i^*|{U}_i^*>0} (y_i;\boldsymbol{\theta})
	=
f_{Y_i^*}(y_i) \, 
\dfrac{\int_{0}^{\infty} f_{U_i^*\vert Y_i^*}(u_i\vert y_i)\,  {\rm d}u_i}{\mathbb{P}(U_i^*>0)}	
	=
f_{Y_i^*}(y_i) \, 
\dfrac{\int_{0}^{\infty} f_{U_i^*\vert Y_i^*}(u_i\vert y_i)\,  {\rm d}u_i}{\int_{0}^{\infty} f_{U_i^*}(u_i) \,  {\rm d}u_i},
\end{equation*}
and by employing identities \eqref{id-1} and \eqref{id-2}, we get
\begin{align*}
f_{{Y}_i^*|{U}_i^*>0} (y_i;\boldsymbol{\theta})
&=
{1\over\sigma_i}\, f_{Z_{1i}} \biggl({y_i-\mu_{1i}\over\sigma_i }\biggr) \,
\dfrac{
\displaystyle	
\int_{0}^{\infty}	
\frac{f_{Z_{1i},\, Z_{2i}} \Big({y_i-\mu_{1i}\over\sigma_i }, {1\over \sqrt{1-\rho_i^2}}\,(u_i- \mu_{2i}) -{\rho_i\over \sqrt{1-\rho_i^2}}\,  ({y_i-\mu_{1i}\over\sigma_i })\Big)
}
{\sqrt{1-\rho_i^2} \, f_{Z_{1i}} ({y_i-\mu_{1i}\over\sigma_i })}
 \,  {\rm d}u_i
}
{
	\displaystyle	
\int_{-\mu_{2i}}^{\infty} 
f_{\rho_i Z_{1i}+\sqrt{1-\rho_i^2}\, Z_{2i}}(u_i) \,  {\rm d}u_i
}
\\[0,2cm]
&=
{1\over\sigma_i}\, f_{Z_{1i}} \biggl({y_i-\mu_{1i}\over\sigma_i }\biggr) \,
\dfrac{
	{\displaystyle	
	\int_{-{1\over \sqrt{1-\rho_i^2}}\,\mu_{2i} -{\rho_i\over \sqrt{1-\rho_i^2}}\,  ({y_i-\mu_{1i}\over\sigma_i })}^{\infty}	}
	{f_{Z_{2i}\vert Z_{1i}} (w_i \vert {y_i-\mu_{1i}\over\sigma_i })
	}
	\,  {\rm d}w_i
}
{
	\displaystyle	
	\int_{-\mu_{2i}}^{\infty} 
	f_{\rho_i Z_{1i}+\sqrt{1-\rho_i^2}\, Z_{2i}}(u_i) \,  {\rm d}u_i
},
\end{align*}
where in the last line a change of variables was used.
Finally, by using the notations of $G_i$ and $H_i$ given in \eqref{int-G} and \eqref{int-H}, respectively, from the above identity and from Proposition \ref{Prop-dual} the proof follows.		
\end{proof}


\begin{corollary}[Gaussian density generator]
If $({Y}^*_i, {U}^*_i)^\top\sim {\rm BSY}(\boldsymbol{\mu},\boldsymbol{\Sigma},g_c)$, where {$g_c(x)=\exp(-x/2)$} is the density generator of the bivariate normal distribution, then the PDF of ${Y}^*_i|{U}^*_i > 0$ is given by
\begin{eqnarray*}
f_{{Y}_i^*|{U}_i^*>0} (y_i;\boldsymbol{\theta})
=
{1\over\sigma_i}\, \phi \biggl({y_i-\mu_{1i}\over\sigma_i }\biggr) \,
{\Phi\Big({1\over \sqrt{1-\rho_i^2}}\,\mu_{2i}+{\rho_i\over \sqrt{1-\rho_i^2}}\,  ({y_i-\mu_{1i}\over\sigma_i })\Big)\over \Phi(\mu_{2i})},
\end{eqnarray*}
wherein $\phi$ and $\Phi$ denote the PDF and CDF of the standard normal  distribution, respectively.
\end{corollary}
\begin{proof}
If $({Y}^*_i, {U}^*_i)^\top$ follows the bivariate normal distribution then there exist independent standard normal random variables $Z_{1i}$ and $Z_{2i}$ such that a stochastic representation of type \eqref{rep-est} is satisfied. Therefore, $Z_{2i}\vert Z_{1i}=z$ and $\rho_i Z_{1i}+\sqrt{1-\rho_i^2}\, Z_{2i}$ are distributed according to the standard normal distribution.
Hence, $G_i$ and $H_i$, given in Proposition \ref{Prop-dual} and Item \eqref{int-H}, respectively, are written as:
\begin{align*}
G_i(x)
=
{\displaystyle	
	\int_{-\infty}^{x}	}
{f_{Z_{2i}\vert Z_{1i}} \Big(w_i\, \Big\vert\, {y_i-\mu_{1i}\over\sigma_i }\Big)
}
\,  {\rm d}w_i
=
\mathbb{P}\Big(Z_{2i}\leqslant -x \,\Big\vert \, Z_{1i} ={y_i-\mu_{1i}\over\sigma_i }\Big)
=
\Phi(x)
\end{align*}
and
\begin{align*}
H_i(x)
=
\int_{-x}^{\infty} 
f_{\rho_i Z_{1i}+\sqrt{1-\rho_i^2}\, Z_{2i}}(u_i) \,  {\rm d}u_i
=
\mathbb{P}\big(\rho_i Z_{1i}+\sqrt{1-\rho_i^2}\, Z_{2i}>-x\big)
=
\Phi(x).
\end{align*}
Applying Theorem \ref{Main-Theorem} the proof concludes.
\end{proof}


{
\begin{remark}\label{remark-t-student}
	It is well-known that, if  $(X_1,X_2)^\top$ is distributed from a bivariate Student-$t$ distribution, notation  $(X_1,X_2)^\top\sim t_\nu(\boldsymbol{\mu},\boldsymbol{\Sigma})$, where $\boldsymbol{\mu}=(\mu_1,\mu_2)\in\mathbb{R}^2$ and $\boldsymbol{\Sigma}$ is as in \eqref{covariance matrix}, then both the marginal and the conditional distributions of $X_2$ given $X_1$ are also univariate Student's $t$ distributions:
	$X_1\sim t_\nu(\mu_1,\sigma_1)$
	and
	\begin{align*}
	X_2\vert X_1=x_1
	\sim
	t_{\nu+1}\biggl(\mu_2+\rho\sigma_2\Big({x_1-\mu_1\over\sigma_1}\Big),{\nu+({x_1-\mu_1\over\sigma_1})^2\over\nu+1}\, \sigma_2^2(1-\rho^2)\biggr).
	\end{align*}
	The above statement is equivalent to
	\begin{align*}
	\sqrt{\nu+1\over (\nu+r^2)(1-\rho^2)}\, 
	\Big({X_2-\mu_2\over \sigma_2}-\rho r\Big)\,\bigg \vert {X_1-\mu_1\over\sigma_1}=r
	\sim
	t_{\nu+1}(0,1)\equiv t_{\nu+1}.
	\end{align*}
\end{remark}
}

\begin{corollary}[Student-$t$ density generator]\label{Student-gen}
If $({Y}^*_i, {U}^*_i)^\top\sim {\rm BSY}(\boldsymbol{\mu},\boldsymbol{\Sigma},g_c)$, where
{$g_c(x)=(1+x/\nu)^{-(\nu+2)/2}$} is the density generator of the bivariate Student-$t$ distribution with $\nu$  degrees of freedom, then the PDF of ${Y}^*_i|{U}^*_i > 0$ is given by
\begin{eqnarray*}
	f_{{Y}_i^*|{U}_i^*>0} (y_i;\boldsymbol{\theta})
	=
	{1\over\sigma_i}\, f_\nu \biggl({y_i-\mu_{1i}\over\sigma_i }\biggr) \,
	{
	F_{\nu+1}\Big(
	\sqrt{{ \nu+1\over \nu+({y_i-\mu_{1i}\over\sigma_i })^2}}\
	\big[{1\over \sqrt{1-\rho_i^2}}\,\mu_{2i}+{\rho_i\over \sqrt{1-\rho_i^2}}\,  ({y_i-\mu_{1i}\over\sigma_i })\big]
	\Big)
	\over 
	F_\nu(\mu_{2i})
},
\end{eqnarray*}
wherein $f_\nu$ and $F_\nu$ denote the PDF and CDF of a classic Student-$t$  distribution with $\nu$  degrees of freedom, respectively.
\end{corollary}
\begin{proof}
It is well-known that the vector $({Y}^*_i, {U}^*_i)^\top$ following  a bivariate Student-$t$ distribution has the stochastic representation; see Subsection 9.2.6, p. 354 of \cite{Balai:09}:
\begin{align}\label{rep-t-student}
\begin{array}{llcc}
&Y_i^*=\sigma_i Z_{1i}+\mu_{1i},
\\[0,4cm]
&U_i^*=\rho_i Z_{1i}+\sqrt{1-\rho_i^2}\, Z_{2i}+\mu_{2i},
\end{array}
\end{align} 
where $Z_{1i}=\sqrt{\nu}\widetilde{Z}_{1i}/\sqrt{Q}$ and $Z_{2i}=\sqrt{\nu}\widetilde{Z}_{2i}/\sqrt{Q}$ are Student-$t$ random variables, wherein $Q\sim \chi^2_\nu$ (chi-square with $\nu$ degrees of freedom) is independent of $\widetilde{Z}_{1i}$ and $\rho_i \widetilde{Z}_{1i}+\sqrt{1-\rho_i^2}\, \widetilde{Z}_{2i}$, and  $\widetilde{Z}_{1i}\stackrel{d}{=}\widetilde{Z}_{2i}\sim N(0,1)$ are independent.

When $({Y}^*_i-\mu_{1i})/\sigma_i= r$, { from Remark \ref{remark-t-student} we have
\begin{align*}
\sqrt{\nu+1\over (\nu+r^2)(1-\rho_i^2)}\, 
({U}^*_i-\mu_{2i}-\rho_i r)\,\bigg \vert {Y_i^*-\mu_{1i}\over\sigma_i}=r
\sim
t_{\nu+1}.
\end{align*}
Hence,}
\begin{align*}
F_{\nu+1}(w)
=
\mathbb{P}
\left(
\sqrt{{\nu+1\over (\nu+r^2)(1-\rho_i^2)}}\, 
({U}^*_i-\mu_{2i}-\rho_i r)
\leqslant
w
\,
\Bigg\vert 
\,
{{Y}^*_i-\mu_{1i}\over\sigma_i}= r
\right), 
\quad -\infty<w<\infty.
\end{align*}
By using the representation \eqref{rep-t-student},
a simple algebraic manipulation shows that the right-hand of the above probability is
\begin{align*}
=
\mathbb{P}
\left(
Z_{2i}\leqslant {w\over \sqrt{{\nu+1\over \nu+r^2}}}
\,
\Bigg\vert 
\,
Z_{1i}=r
\right).
\end{align*}
Setting $w=-\sqrt{{(\nu+1)/(\nu+r^2)}} x$ and $r={(y_i-\mu_{1i})/\sigma_i }$ we obtain
\begin{align*}
\mathbb{P}\Big(Z_{2i}\leqslant -x \,\Big\vert \, Z_{1i} =
{y_i-\mu_{1i}\over\sigma_i }\Big)
=
F_{\nu+1}\left(\sqrt{{\nu+1\over \nu+({y_i-\mu_{1i}\over\sigma_i })^2}}\ x\right).
\end{align*}

Therefore, the function $G_i$, given in Proposition \ref{Prop-dual}, is
\begin{align*}
G_i(x)
=
{\displaystyle	
	\int_{-\infty}^{x}	}
{f_{Z_{2i}\vert Z_{1i}} \Big(w_i\, \Big\vert\, {y_i-\mu_{1i}\over\sigma_i }\Big)
}
\,  {\rm d}w_i
&=
\mathbb{P}\Big(Z_{2i}\leqslant -x \,\Big\vert \, Z_{1i} =
{y_i-\mu_{1i}\over\sigma_i }\Big)
\\[0,2cm]
&=
F_{\nu+1}\left(\sqrt{{\nu+1\over \nu+({y_i-\mu_{1i}\over\sigma_i })^2}}\ x\right).
\end{align*}

On the other hand, note that $\rho_i Z_{1i}+\sqrt{1-\rho_i^2}\, Z_{2i}=\sqrt{\nu}(\rho_i \widetilde{Z}_{1i}+\sqrt{1-\rho_i^2}\, \widetilde{Z}_{2i})/\sqrt{Q}$ is distributed according to the a Student-$t$ distribution with $\nu$ degrees of freedom because $\rho_i \widetilde{Z}_{1i}+\sqrt{1-\rho_i^2}\, \widetilde{Z}_{2i}\sim N(0,1)$.
Then
\begin{align*}
H_i(x)
=
\int_{-x}^{\infty} 
f_{\rho_i Z_{1i}+\sqrt{1-\rho_i^2}\, Z_{2i}}(u_i) \,  {\rm d}u_i
=
\mathbb{P}\big(\rho_i Z_{1i}+\sqrt{1-\rho_i^2}\, Z_{2i}>-x\big)
=
F_{\nu}(x).
\end{align*}
Applying Theorem \ref{Main-Theorem} the proof follows.
\end{proof}


Given the density generator $g_c$ we can directly determine the PDF of ${Y}^*_i|{U}^*_i > 0$ as follows.
\begin{corollary}
	If $({Y}^*_i, {U}^*_i)^\top\sim {\rm BSY}(\boldsymbol{\mu},\boldsymbol{\Sigma},g_c)$ then the PDF of ${Y}^*_i|{U}^*_i > 0$ is given by
	\begin{eqnarray*}
		f_{{Y}_i^*|{U}_i^*>0} (y_i;\boldsymbol{\theta})
		=
		{\rho_i \sqrt{1-\rho_i^2}\over\sigma_i}
		\,
		{	
			{
				{
					\int_{-\infty}^{z_i} g_c(({y_i-\mu_{1i}\over\sigma_i })^2+w_i^2)\,  {\rm d}w_i
				}
			}
			\over 
			\int_{-\mu_{2i}}^{\infty} 
			\big\{
			\int_{-\infty}^{\infty} 
			{g_c\big(({x\over \rho_i})^2+\big({u_i-x\over \sqrt{1-\rho_i^2}}\big)^2\big)}
			\,{\rm d}x
			\big\}
			\,  {\rm d}u_i
		},
	\end{eqnarray*}
	where $z_i=\mu_{2i}/\sqrt{1-\rho_i^2}+
	{\rho_i}\,  ({y_i-\mu_{1i}\over\sigma_i })/\sqrt{1-\rho_i^2}$.
\end{corollary}
\begin{proof}
	By using the known formula for the PDF of sum of two random variables $X$ and $Y$:
	\begin{align*}
	f_{X+Y}(z)=\int_{-\infty}^{\infty} f_{X,Y}(x,z-x)\,{\rm d}x,
	\end{align*}
	we have that $H_i$, defined in \eqref{int-H}, can be written as
	\begin{align*}
	H_i(x)
	=
	{1\over  \rho_i \sqrt{1-\rho_i^2}}
	{
		\int_{-x}^{\infty} 
		\big\{
		\int_{-\infty}^{\infty} 
		{g_c\big(({\zeta\over \rho_i})^2+\big({u_i-\zeta\over \sqrt{1-\rho_i^2}}\big)^2\big)}
		\,{\rm d}\zeta
		\big\}
		\,  {\rm d}u_i
		\over
		\pi  \int_{0}^{\infty} g_c(u)\, {\rm d}{u}
	}.
	\end{align*}

	Proposition \ref{prop-marg-sym} provides
	\begin{align}\label{density-Z-1}
	f_{Z_{1i}} \Big({y_i-\mu_{1i}\over\sigma_i }\Big)
	=
	{\int_{-\infty}^{\infty}g_c(({y_i-\mu_{1i}\over\sigma_i })^2+y^2) \, {\rm d}y\over \pi \int_{0}^{\infty} g_c(u)\, {\rm d}{u}}.
	\end{align}
	
	By combining Proposition \ref{Prop-dual}, Equations  \eqref{eq-joint-pdf} and \eqref{density-Z-1}, we have
	\begin{align*}
	G_i(x)
	=
	{
		\int_{-\infty}^{x}	
		{f_{Z_{1i}, Z_{2i}}({y_i-\mu_{1i}\over\sigma_i }, w_i)
		}
		\,  {\rm d}w_i	
		\over
		f_{Z_{1i}}({y_i-\mu_{1i}\over\sigma_i })
	}
	=
	{
		\int_{-\infty}^{x}	
		g_c(({y_i-\mu_{1i}\over\sigma_i })^2+ w_i^2)
		\,  {\rm d}w_i	
		\over 
		{\int_{-\infty}^{\infty}g_c(({y_i-\mu_{1i}\over\sigma_i })^2+y^2) \, {\rm d}y}
	}.
	\end{align*}
	
	Substituting the above formulas of $H_i$, $f_{Z_{1i}}$ and $G_i$ into Theorem \ref{Main-Theorem}, we complete the proof.
\end{proof}

\subsection{Maximum likelihood estimation}
\label{Maximum likelihood estimation}
By combining Equation \eqref{eq:pdfy} with Theorem \ref{Main-Theorem}, the following formula for the PDF of $Y_i$ is valid:
\begin{align*}
f_{Y_i}(y_i;\boldsymbol{\theta})
&=
(1-H_i(\mu_{2i}))^{1-u_i}
(H_i(\mu_{2i}))^{u_i}
\Biggl[
{1\over\sigma_i}\, f_{Z_{1i}} \biggl({y_i-\mu_{1i}\over\sigma_i }\biggr) \,
{G_i\big(\tau_i+\alpha_i  ({y_i-\mu_{1i}\over\sigma_i })\big)\over H_i(\mu_{2i})}
\Biggr]^{u_i},
\end{align*}
where $\alpha_i=\rho_i/\sqrt{1-\rho_i^2}$,  $\tau_i=\mu_{2i}/\sqrt{1-\rho_i^2}$,
$u_i = 1$ if $u_i^{*}>0$ and $u_i = 0$ otherwise, $
g_1(\mu_{1i}) = \boldsymbol{x}_i^\top \boldsymbol{\beta}  $,
$g_2(\mu_{2i}) = \boldsymbol{w}_i^\top \boldsymbol{\gamma} $,
$h_1(\sigma_{i}) = \boldsymbol{z}_i^\top \boldsymbol{\lambda}$ and
$h_2(\rho_{i}) = \boldsymbol{v}_i^\top \boldsymbol{\kappa}$.

The log-likelihood of the symmetric generalized Heckman model for $\boldsymbol{\theta} = (\boldsymbol{\beta}^\top, \boldsymbol{\gamma}^\top, \boldsymbol{\lambda}^\top, \boldsymbol{\kappa}^\top)^{\top}$ is given by
\begin{align}\label{eq:loglik_symmetric}
\ell(\boldsymbol{\theta})
&=
\sum_{i=1}^{n} \log f_{Y_i}(y_i;\boldsymbol{\theta})
\nonumber
\\[0,2cm]
&=
\sum_{i=1}^{n}
u_i
\biggl[
-\log(\sigma_i)
+
\log
f_{Z_{1i}}\left(\frac{y_i - \mu_{1i}}{\sigma_i}\right)
+
\log
{G_{i}\left(\tau_i+\alpha_i 
\left(\frac{y_i - \mu_{1i}}{\sigma_i}\right)\right)}
\biggr]
\nonumber
\\[0,2cm]
&
+
\sum_{i=1}^{n} (1 - u_i)\log(1-H_i(\mu_{2i})).
\end{align}
To obtain the ML estimate of $\boldsymbol{\theta}$, we maximize the log-likelihood function \eqref{eq:loglik_symmetric} by equating the score vector $\dot{\ell}(\boldsymbol{\theta})$ to zero, providing the likelihood equations. They are solved by means of an iterative procedure for non-linear optimization, such as the Broyden-Fletcher-Goldfarb-Shanno (BFGS) quasi-Newton method.

The likelihood equations are given by
{\small
\begin{align*}
	0={\partial \ell(\boldsymbol{\theta}) \over \partial\beta_j}
	&=
	\sum_{i=1}^{n}
	{u_i\over\sigma_i}
	\Biggl[-
{
	f'_{Z_{1i}}(\frac{y_i - \mu_{1i}}{\sigma_i})
	\over
	f_{Z_{1i}}(\frac{y_i - \mu_{1i}}{\sigma_i})
}
-
\frac{\alpha_i 
f_{Z_{2i}\vert Z_{1i}} \big(\alpha_i\, (\frac{y_i - \mu_{1i}}{\sigma_i}) + \tau_i\, \big\vert\, {y_i-\mu_{1i}\over\sigma_i }\big)
}
{G_i \big(\alpha_i\, (\frac{y_i - \mu_{1i}}{\sigma_i}) + \tau_i\big)}
\Biggr] {\partial \mu_{1i} \over \partial\beta_j};
	\\[0,2cm]
	0={\partial \ell(\boldsymbol{\theta}) \over \partial\gamma_r}
	&=
	\sum_{i=1}^{n}
	\left[
	{u_i\over \sqrt{1-\rho_i^2}}\,
	{	
		f_{Z_{2i}\vert Z_{1i}} \big(\alpha_i\, (\frac{y_i - \mu_{1i}}{\sigma_i}) + \tau_i\, \big\vert\, {y_i-\mu_{1i}\over\sigma_i }\big)
		\over 
G_i \big(\alpha_i\, (\frac{y_i - \mu_{1i}}{\sigma_i}) + \tau_i\big)	
}
	+
	{f_{\rho_i Z_{1i}+\sqrt{1-\rho_i^2}\, Z_{2i}}(-\mu_{2i}) \over  1-H_i(\mu_{2i})}
	\right]
	{\partial \mu_{2i} \over \partial\gamma_r};
	\\[0,2cm]
	0={\partial \ell(\boldsymbol{\theta}) \over \partial\lambda_s}
	&=
	\sum_{i=1}^{n}
	{u_i\over\sigma_i}
	\Biggl[
	-1
	-
	{(y_i-\mu_{1i})\over\sigma_i}\,
{
	f'_{Z_{1i}}(\frac{y_i - \mu_{1i}}{\sigma_i})
	\over
	f_{Z_{1i}}(\frac{y_i - \mu_{1i}}{\sigma_i})
}
	-
	\alpha_i\, \frac{(y_i - \mu_{1i})}{\sigma_i}\,
\frac{
	f_{Z_{2i}\vert Z_{1i}} \big(\alpha_i\, (\frac{y_i - \mu_{1i}}{\sigma_i}) + \tau_i\, \big\vert\, {y_i-\mu_{1i}\over\sigma_i }\big)
}
{G_i \big(\alpha_i\, (\frac{y_i - \mu_{1i}}{\sigma_i}) + \tau_i\big)}
	\Biggr]
	{\partial \sigma_{i} \over \partial\lambda_s};
	\\[0,2cm]
	0={\partial \ell(\boldsymbol{\theta}) \over \partial\kappa_m}
	&=
	\sum_{i=1}^{n}	
{u_i\over \sqrt{1-\rho_i^2}}\,
	\Biggl[
{({y_i - \mu_{1i}\over \sigma_i})\over 1-\rho_i^2}
-
\mu_{2i}\rho_i
\Biggr]
\frac{
	f_{Z_{2i}\vert Z_{1i}} \big(\alpha_i\, (\frac{y_i - \mu_{1i}}{\sigma_i}) + \tau_i\, \big\vert\, {y_i-\mu_{1i}\over\sigma_i }\big)
}
{G_i \big(\alpha_i\, (\frac{y_i - \mu_{1i}}{\sigma_i}) + \tau_i\big)}\,
	{\partial\rho_i \over\partial\kappa_m};
\end{align*}
}\noindent
where
\begin{align*}
{\partial \mu_{1i} \over \partial\beta_j}
&=
{x_{ij}\over g_1'(\mu_{1i})}, \quad j=1,\ldots,k; \ i=1,\ldots,n;
\\[0,2cm]
{\partial \mu_{2i} \over \partial\gamma_r}
&=
{w_{ir}\over g_2'(\mu_{2i})}, \quad r=1,\ldots,l; \ i=1,\ldots,n;
\\[0,2cm]
{\partial \sigma_{i} \over \partial\lambda_s}
&=
{z_{is}\over h_1'(\sigma_i)}, \quad s=1,\ldots,p; \ i=1,\ldots,n;
\\[0,2cm]
{\partial\rho_i \over\partial\kappa_m}
&=
{v_{im}\over h_2'(\rho_i)}, \quad m=1,\ldots,q; \ i=1,\ldots,n.
\end{align*}

\section{Generalized Heckman-Student-$t$ model}\label{sec:3}

The generalized Heckman-normal model proposed by \cite{Bastos2021} is a special case of \eqref{errors:gen} when the underlying distribution is bivariate normal. In this work, we focus on the generalized Heckman-$t$ model, which is based on the bivariate Student-$t$ (B$t$) distribution. This distribution is a good alternative in the symmetric family of distributions because it possesses has heavier tails than the bivariate normal distribution. From \eqref{eq:pdf:sym}, if $\boldsymbol{Y}=(Y_1,Y_2)^{\top}$ follows a B$t$ distribution, then the associated PDF is given by
\begin{eqnarray}\label{eq:EST}
f(\boldsymbol{y}; \boldsymbol{\mu};  \boldsymbol{\Sigma}, \nu) 
=
{\frac{1}{|\boldsymbol{\Sigma}|^{1/2} Z_{g_c}}}
\left\{ 1 + \frac{(\boldsymbol{y} - \boldsymbol{\mu})^\top \boldsymbol{\Sigma} ^{-1}(\boldsymbol{y} - \boldsymbol{\mu})}{\nu} \right\}^{-(\nu+2)/2},
\end{eqnarray}
where $\nu$ is the number of degrees of freedom. Here, the density generator of the B$t$ distribution is given by $g_c(x)=(1+x/\nu)^{-(\nu+2)/2}$ {, $|\boldsymbol{\Sigma}|=\sigma_i^2(1-\rho_i^2)$ and $Z_{g_c}= [\nu\pi\Gamma(\nu/2)]/\Gamma((\nu+2)/2)$ is a normalization constant.} Therefore, if $({Y}^*_i, {U}^*_i)$ follow a B$t$ distribution, then, by Corollary \ref{Student-gen}, the PDF of ${Y}^*_i|{U}^*_i>0$ is written as
\begin{eqnarray*}
	f_{{Y}_i^*|{U}_i^*>0} (y_i;\mu_{1i},\sigma^2_i,\alpha_i,\tau_i,\nu)
	=
	{1\over\sigma_i}\, f_\nu \biggl({y_i-\mu_{1i}\over\sigma_i }\biggr) \,
	\frac{
		F_{\nu+1}\Big(
		\sqrt{{\nu+1\over \nu+({y_i-\mu_{1i}\over\sigma_i })^2}}\ 
		\big(\tau_i+\alpha_i  ({y_i-\mu_{1i}\over\sigma_i })\big)
		\Big)
	}{
		F_\nu\big(\tau_i/\sqrt{1+\alpha_i^2}\, \big)
	},
\end{eqnarray*}
where $f_{\nu}$ and $F_\nu$ are the PDF and CDF, respectively, of a univariate Student-$t$ distribution with $\nu$ degrees of freedom,
$\alpha_i=\rho_i/\sqrt{1-\rho_i^2}$ and  $\tau_i=\mu_{2i}/\sqrt{1-\rho_i^2}$.
The log-likelihood for $\boldsymbol{\theta} = (\boldsymbol{\beta}^\top, \boldsymbol{\gamma}^\top, \boldsymbol{\lambda}^\top, \boldsymbol{\kappa}^\top,\nu)^{\top}$ is given by
\begin{align}\label{eq:loglik_student}
\ell(\boldsymbol{\theta})
&=
\sum_{i=1}^{n} \log f_{Y_i}(y_i;\boldsymbol{\theta})
\nonumber
\\[0,2cm]
&=
\sum_{i=1}^{n} 
\log\left\{
(F_\nu(-\mu_{2i}))^{1-u_i}
(F_\nu(\mu_{2i}))^{u_i}
(
f_{{Y}_i^*|{U}_i^*>0} (y_i;\mu_{1i},\sigma^2_i,\alpha_i,\tau_i,\nu)
)^{u_i}
\right\}
\nonumber
\\[0,3cm]
&=
\sum_{i=1}^{n}
u_i
\left[
-\log(\sigma_i)
+
\log
f_{\nu}\left(\frac{y_i - \mu_{1i}}{\sigma_i}\right)
+
\log
		F_{\nu+1}\left(
\sqrt{{\nu+1\over \nu+({y_i-\mu_{1i}\over\sigma_i })^2}}\ 
\big(\tau_i+\alpha_i  \big({y_i-\mu_{1i}\over\sigma_i }\big)\big)
\right)
\right]
\nonumber
\\[0,2cm]
&+
\sum_{i=1}^{n} (1 - u_i)\log F_\nu(-\mu_{2i}),
\end{align}
where $u_i = 1$ if $u_i^{*}>0$ and $u_i = 0$ otherwise, $\mu_{1i}$, $\mu_{2i}$, $\sigma_{i}$ and $\rho_{i}$ are as in \eqref{errors:gen}. The ML estimate of $\boldsymbol{\theta}$ is obtained by maximizing the log-likelihood function \eqref{eq:loglik_student}, that is, by equating the score vector $\dot{\ell}(\boldsymbol{\theta})$ (given in Subsection \ref{Maximum likelihood estimation}) to zero, providing the likelihood equations. They are solved using an iterative procedure for non-linear optimization, such as the BFGS quasi-Newton method.

\section{Monte Carlo simulation}\label{sec:4}
\noindent
In this section, we carry out Monte Carlo simulation studies to evaluate the performance of the ML estimators under the symmetric generalized Heckman model. We focus on the generalized Heckman-$t$ model and consider three different set of true parameter value, which leads to scenarios covering moderate to high censoring percentages. The studies consider simulated data generated from each scenario according to
\begin{equation}\label{eq:mc}
  \mu_{1i} = \beta_1 + \beta_2x_{1i} + \beta_3x_{2i},
\end{equation}
\begin{equation}
 \mu_{2i} = \gamma_1 + \gamma_2 x_{1i} + \gamma_3 x_{2i} + \gamma_4 x_{3i},
\end{equation}
\begin{equation}
    \log \sigma_i = \lambda_1 + \lambda_2 x_{1i},
\end{equation}
\begin{equation}\label{eq:mc2}
    \text{arctanh}\, \rho_i = \kappa_1 + \kappa_2 x_{1i},
\end{equation}
for $i = 1, \ldots, n$, $x_{1i}$, $x_{2i}$ and $x_{3i}$ are covariates obtained from a normal distribution in the interval (0,1). Moreover, the simulation scenarios consider sample size $n \in \{ 500, 1000, 2000\}$ and $\nu=4$, with $\text{NREP}=1000$ Monte Carlo replicates for each sample size. In the structure presented in \eqref{eq:mc} - \eqref{eq:mc2}, $\mu_{1i}$ is the primary interest equation, while $\mu_{2i}$ represents the selection equation.  The R software has been used to do all numerical calculations; see \cite{rmanual}.

The performance of the ML estimators are evaluated through the bias and mean squared error (MSE), computed from the Monte Carlo replicas as
\begin{equation}
 \widehat{\textrm{Bias}}(\widehat{\theta}) = \frac{1}{\text{NREP}} \sum_{i = 1}^{\text{NREP}} \widehat{\theta}^{(i)} - \theta
\quad \text{and}\quad
\widehat{\mathrm{MSE}}(\widehat{\theta}) = \frac{1}{\text{NREP}} \sum_{i = 1}^{\text{NREP}} (\widehat{\theta}^{(i)} - \theta)^2,
\end{equation}
where $\theta$ and $\widehat{\theta}^{(i)}$ are the true parameter value and its respective $i$-th ML estimate, and $\text{NREP}$ is the number of Monte Carlo replicas.

We consider the following sets of true parameter values for the regression structure in \eqref{eq:mc}-\eqref{eq:mc2}:
\begin{itemize}
 \item Scenario 1) $\boldsymbol\beta= (1.1, 0.7, 0.1)^{\top}$, $\boldsymbol\gamma = (0.9, 0.5, 1.1, 0.6)^{\top}$, and $\boldsymbol\lambda = (-0.4, 0.7)^{\top}$ and $\boldsymbol\kappa = (0.3, 0.5)^{\top}$.  
\item  Scenario 2) $\boldsymbol\beta= (1.0, 0.7, 1.1)^{\top}$, $\boldsymbol\gamma = (0.9, 0.5, 1.1, 0.6)^{\top}$, $\boldsymbol\lambda = (-0.2, 1.2)^{\top}$, and $\boldsymbol\kappa = (0.7, 0.3)^{\top}$ or $\boldsymbol\kappa = (-0.7, 0.3)^{\top}$. 
\item Scenario 3) $\boldsymbol\beta= (1.1, 0.7, 0.1)^{\top}$, $\boldsymbol\gamma = (0, 0.5, 1.1, 0.6)^{\top}$, $\boldsymbol\lambda = (-0.4, 1.2)^{\top}$, and $\boldsymbol\kappa = (-0.3, -0.3)^{\top}$ (moderate correlation) or $\boldsymbol\kappa = (-0.7, -0.7)^{\top}$ (strong correlation).
\end{itemize}

To keep the censoring proportion around 50\%, in Scenario 1 a threshold greater than zero was used, so $U_{i}^{*} > a$. According to \cite{Bastos2021}, in general the value of $a$ is zero, as any other value would be absorbed by the intercept, so considering another value does not cause problems for the model. In Scenario 2, the dispersion and correlation parameters were changed and the censoring proportion was maintained around 30\%. In Scenario 3, the censoring rate around 50\% was obtained by changing the parameters of the selection equation $\mu_{2i}$.

The ML estimation results for the Scenarios 1), 2) and 3) are presented in Tables \ref{tab:mc1}-\ref{tab:mc3}, respectively, wherein the bias and MSE are all reported. As the ML estimators are consistent and asymptotically normally
distributed, we expect the bias and MSE to approach zero as $n$ grows. Moreover, we expect that the performances
of the estimates deteriorate as the censoring proportion (\%) grows. A look at the results in Tables \ref{tab:mc1}-\ref{tab:mc3} allows us to conclude that, as the sample size increases, the bias and MSE both decrease, as expected. In addition, the performances of the estimates decrease when the censoring proportion increases.

\def\tablename{Table}
\begin{table}[h!]
\caption{Bias and MSE for the indicated ML estimates of the generalized Heckman-$t$ model parameters (Scenario 1).}
\centering
\footnotesize
\label{tab:mc1}
\begin{tabular}{lccccccc}
\hline
                                  &      &                 & \multicolumn{2}{c}{Generalized Heckman-$t$} &                 & \multicolumn{2}{c}{Generalized Heckman-$t$} \\

Parameters                        & n    & Censoring & Bias                 & MSE                & Censoring & Bias                 & MSE                \\
\hline
\multirow{3}{*}{$\beta_1$ \ \ \ \ 1.1}    & 500  & 30.8878         & 0.0034               & 0.0053             & 54.0284         & 0.0082               & 0.0148             \\
                                  & 1000 & 30.9352         & 0.0019               & 0.0025             & 53.421          & 0.0016               & 0.0060             \\
                                  & 2000 & 31.0035         & 0.0021               & 0.0011             & 52.8706         & 0.0058               & 0.0030             \\
\hline
\multirow{3}{*}{$\beta_2$ \ \ \ \ 0.7}    & 500  & 30.8878         & 0.0037               & 0.0014             & 54.0284         & 0.0060               & 0.0034             \\
                                  & 1000 & 30.9352         & 0.0020               & 0.0007             & 53.421          & 0.0033               & 0.0013             \\
                                  & 2000 & 31.0035         & 0.0016               & 0.0003             & 52.8706         & 0.0020               & 0.0005             \\
\hline
\multirow{3}{*}{$\beta_3$ \ \ \ \ 0.1}    & 500  & 30.8878         & 0.0009               & 0.002              & 54.0284         & -0.0014              & 0.0045             \\
                                  & 1000 & 30.9352         & -0.0003              & 0.0008             & 53.421          & 0.0003               & 0.0017             \\
                                  & 2000 & 31.0035         & 0.0008               & 0.0004             & 52.8706         & -0.0022              & 0.0008             \\
\hline
\multirow{3}{*}{$\gamma_1$ \ \ \ \ 0.9}   & 500  & 30.8878         & 0.0163               & 0.0137             & 54.0284         & -1.0167              & 1.0443             \\
                                  & 1000 & 30.9352         & 0.0055               & 0.0065             & 53.421          & -1.0095              & 1.0246             \\
                                  & 2000 & 31.0035         & -0.0008              & 0.0032             & 52.8706         & -0.9981              & 0.9980             \\
\hline
\multirow{3}{*}{$\gamma_2$ \ \ \ \ 0.5}   & 500  & 30.8878         & 0.0066               & 0.0091             & 54.0284         & 0.0071               & 0.0077             \\
                                  & 1000 & 30.9352         & 0.0064               & 0.0044             & 53.421          & 0.0026               & 0.0039             \\
                                  & 2000 & 31.0035         & 0.0025               & 0.0021             & 52.8706         & 0.0021               & 0.0019             \\
\hline
\multirow{3}{*}{$\gamma_3$ \ \ \ \ 1.1}   & 500  & 30.8878         & 0.0177               & 0.0194             & 54.0284         & 0.0131               & 0.0170             \\
                                  & 1000 & 30.9352         & 0.0081               & 0.0085             & 53.421          & 0.0090               & 0.0074             \\
                                  & 2000 & 31.0035         & 0.0032               & 0.0041             & 52.8706         & 0.0044               & 0.0038             \\
\hline
\multirow{3}{*}{$\gamma_4$ \ \ \ \ 0.6}   & 500  & 30.8878         & 0.0091               & 0.0110             & 54.0284         & 0.0076               & 0.0089             \\
                                  & 1000 & 30.9352         & 0.0081               & 0.0085             & 53.421          & 0.0043               & 0.0045             \\
                                  & 2000 & 31.0035         & 0.0028               & 0.0026             & 52.8706         & 0.0020               & 0.0020             \\
\hline
\multirow{3}{*}{$\kappa_1$ \ \ \ \ 0.3}   & 500  & 30.8878         & 0.0139               & 0.0564             & 54.0284         & 0.0034               & 0.0467             \\
                                  & 1000 & 30.9352         & 0.0048               & 0.0048             & 53.421          & 0.0080               & 0.0207             \\
                                  & 2000 & 31.0035         & -0.0005              & 0.0102             & 52.8706         & -0.0029              & 0.0097             \\
\hline
\multirow{3}{*}{$\kappa_2$ \ \ \ \ 0.5}   & 500  & 30.8878         & 0.0589               & 0.0554             & 54.0284         & 0.0416               & 0.0383             \\
                                  & 1000 & 30.9352         & 0.0054               & 0.0225             & 53.421          & 0.0202               & 0.0157             \\
                                  & 2000 & 31.0035         & 0.0120               & 0.0083             & 52.8706         & 0.0136               & 0.0064             \\
\hline
\multirow{3}{*}{$\lambda_1$ \ \ \ \ -0.4} & 500  & 30.8878         & 0.0011               & 0.0049             & 54.0284         & -0.0029              & 0.0075             \\
                                  & 1000 & 30.9352         & 0.0210               & 0.018              & 53.421          & 0.0009               & 0.0035             \\
                                  & 2000 & 31.0035         & 0.0018               & 0.0011             & 52.8706         & -0.0014              & 0.0018             \\
\hline
\multirow{3}{*}{$\lambda_2$ \ \ \ \ 0.7}  & 500  & 30.8878         & 0.0017               & 0.0031             & 54.0284         & 0.0037               & 0.0046             \\
                                  & 1000 & 30.9352         & 0.0006               & 0.0023             & 53.421          & 0.0026               & 0.0021             \\
                                  & 2000 & 31.0035         & 0.0007               & 0.0007             & 52.8706         & 0.0006               & 0.0010             \\
\hline
\multirow{3}{*}{$\nu$ \ \ \ \ \ 4 }               & 500  & 30.8878         & 0.3633               & 2.3387             & 54.0284         & 0.6463               & 8.4442             \\
                                  & 1000 & 30.9352         & 0.1506               & 0.5565             & 53.421          & 0.2221               & 0.8141             \\
                                  & 2000 & 31.0035         & 0.0833               & 0.2377             & 52.8706         & 0.0820               & 0.3078
   \\
\hline
\end{tabular}
\end{table}

\def\tablename{Table}
\begin{table}[h!]
\caption{Bias and MSE for the indicated ML estimates of the generalized Heckman-$t$ model parameters (Scenario 2).}
\centering
\footnotesize
\label{tab:mc2}
\begin{tabular}{llccclcccl}
\hline

     &                                  & \multicolumn{1}{c}{}                & \multicolumn{2}{c}{Generalized Heckman-$t$}          & \multicolumn{1}{c}{}             & \multicolumn{1}{c}{}                & \multicolumn{2}{c}{Generalized Heckman-$t$}          \\
n    & Parameters                       & \multicolumn{1}{c}{Censoring} & \multicolumn{1}{c}{Bias} & \multicolumn{1}{c}{MSE} & Parameters                       & \multicolumn{1}{c}{Censoring} & \multicolumn{1}{c}{Bias} & \multicolumn{1}{c}{MSE} \\
\hline
500  & \multirow{3}{*}{$\beta_1$ \ \ \ \ 1.0}   & 31.017                              & 0.0048                   & 0.0045                  & \multirow{3}{*}{$\beta_1$ \ \ \ \ 1.0}   & 30.9268                             & 0.0052                   & 0.0039                  \\
1000 &                                  & 30.9395                             & 0.0031                   & 0.0022                  &                                  & 30.8644                             & -0.0002                  & 0.0016                  \\
2000 &                                  & 30.9098                             & 0.0013                   & 0.0009                  &                                  & 30.9734                             & 0.0005                   & 0.0009                  \\
\hline
500  & \multirow{3}{*}{$\beta_2$ \ \ \ \ 0.7}   & 31.017                              & 0.0013                   & 0.0009                  & \multirow{3}{*}{$\beta_2$  \ \ \ \ 0.7}   & 30.9268                             & 0.0045                   & 0.0008                  \\
1000 &                                  & 30.9395                             & -0.0002                  & 0.0046                  &                                  & 30.8644                             & 0.0010                    & 0.0003                  \\
2000 &                                  & 30.9098                             & 0.0002                   & 0.0001                  &                                  & 30.9734                             & 0.0006                   & 0.0001                  \\
\hline
500  & \multirow{3}{*}{$\beta_3$ \ \ \ \ 1.1}   & 31.017                              & -0.001                   & 0.0011                  & \multirow{3}{*}{$\beta_3$ \ \ \ \ 1.1}   & 30.9268                             & 0.0001                   & 0.0007                  \\
1000 &                                  & 30.9395                             & 0.0015                   & 0.0008                  &                                  & 30.8644                             & 0.0009                   & 0.0003                  \\
2000 &                                  & 30.9098                             & -0.0007                  & 0.0002                  &                                  & 30.9734                             & 0.0003                   & 0.0001                  \\
\hline

500  & \multirow{3}{*}{$\gamma_1$ \ \ \ \ 0.9}  & 31.017                              & 0.0071                   & 0.0141                  & \multirow{3}{*}{$\gamma_1$ \ \ \ \ 0.9}  & 30.9268                             & 0.0108                   & 0.0153                  \\
1000 &                                  & 30.9395                             & 0.0165                   & 0.0814                  &                                  & 30.8644                             & 0.0205                   & 0.0555                  \\
2000 &                                  & 30.9098                             & 0.0036                   & 0.0031                  &                                  & 30.9734                             & 0.0027                   & 0.0034                  \\
\hline
500  & \multirow{3}{*}{$\gamma_2$ \ \ \ \ 0.5}  & 31.017                              & 0.0072                   & 0.0090                   & \multirow{3}{*}{$\gamma_2$ \ \ \ \ 0.5}  & 30.9268                             & 0.0111                   & 0.0096                  \\
1000 &                                  & 30.9395                             & 0.0076                   & 0.0101                  &                                  & 30.8644                             & 0.0115                   & 0.0270                   \\
2000 &                                  & 30.9098                             & -0.0005                  & 0.0020                   &                                  & 30.9734                             & 0.0016                   & 0.0035                  \\
\hline
500  & \multirow{3}{*}{$\gamma_3$ \ \ \ \ 1.1}  & 31.017                              & 0.013 0                   & 0.0180                   & \multirow{3}{*}{$\gamma_3$ \ \ \ \ 1.1}  & 30.9268                             & 0.0143                   & 0.0175                  \\
1000 &                                  & 30.9395                             & 0.0151                   & 0.0254                  &                                  & 30.8644                             & 0.0170                    & 0.0450                   \\
2000 &                                  & 30.9098                             & 0.0020                    & 0.0041                  &                                  & 30.9734                             & 0.0028                   & 0.0051                  \\
\hline
500  & \multirow{3}{*}{$\gamma_4$ \ \ \ \ 0.6}  & 31.017                              & 0.0067                   & 0.0102                  & \multirow{3}{*}{$\gamma_4$ \ \ \ \ 0.6}  & 30.9268                             & -0.0003                  & 0.0094                  \\
1000 &                                  & 30.9395                             & 0.0064                   & 0.0195                  &                                  & 30.8644                             & 0.0077                   & 0.0139                  \\
2000 &                                  & 30.9098                             & 0.001                    & 0.0023                  &                                  & 30.9734                             & 0.0023                   & 0.0024                  \\
\hline
500  & \multirow{3}{*}{$\kappa_1$ \ \ \ \ 0.7}  & 31.017                              & 0.0299                   & 0.0560                  & \multirow{3}{*}{$\kappa_1$ \ \ \ \ -0.7}  & 30.9268                             & -0.0449                  & 0.0662                  \\
1000 &                                  & 30.9395                             & 0.0197                   & 0.0628                  &                                  & 30.8644                             & -0.0198                  & 0.033                   \\
2000 &                                  & 30.9098                             & 0.0031                   & 0.0089                  &                                  & 30.9734                             & -0.0058                  & 0.0218                  \\
\hline
500  & \multirow{3}{*}{$\kappa_2$ \ \ \ \ 0.3}  & 31.017                              & 0.052                    & 0.0702                  & \multirow{3}{*}{$\kappa_2$ \ \ \ \ 0.3}  & 30.9268                             & 0.0363                   & 0.067                   \\
1000 &                                  & 30.9395                             & 0.0197                   & 0.0576                  &                                  & 30.8644                             & 0.0141                   & 0.0632                  \\
2000 &                                  & 30.9098                             & 0.0097                   & 0.0097                  &                                  & 30.9734                             & 0.0108                   & 0.0345                  \\
\hline
500  & \multirow{3}{*}{$\lambda_1$ \ \ \ \ -0.2} & 31.017                              & 0.0004                   & 0.0053                  & \multirow{3}{*}{$\lambda_1$ \ \ \ \ -0.2} & 30.9268                             & 0.0043                   & 0.0054                  \\
1000 &                                  & 30.9395                             & -0.0018                  & 0.0041                  &                                  & 30.8644                             & -0.0031                  & 0.0039                  \\
2000 &                                  & 30.9098                             & 0.0019                   & 0.0012                  &                                  & 30.9734                             & -0.0004                  & 0.0013                  \\
\hline
500  & \multirow{3}{*}{$\lambda_2$ \ \ \ \ 1.2} & 31.017                              & 0.0032                   & 0.0038                  & \multirow{3}{*}{$\lambda_2$ \ \ \ \ 1.2} & 30.9268                             & 0.0096                   & 0.0035                  \\
1000 &                                  & 30.9395                             & 0.0018                   & 0.0020                   &                                  & 30.8644                             & 0.0058                   & 0.0016                  \\
2000 &                                  & 30.9098                             & -0.0001                  & 0.0007                  &                                  & 30.9734                             & 0.0024                   & 0.0010                   \\
\hline
500  & \multirow{3}{*}{$\nu$ \ \ \ \ \ 4 }               & 31.017                              & 0.4404                   & 2.3837                  & \multirow{3}{*}{$\nu$ \ \ \ \ \ 4 }               & 30.9268                             & 0.5096                   & 17.9132                 \\
1000 &                                  & 30.9395                             & 0.1418                   & 0.6321                  &                                  & 30.8644                             & 0.1619                   & 0.6922                  \\
2000 &                                  & 30.9098                             & 0.1234                   & 0.2899                  &                                  & 30.9734                             & 0.0789                   & 0.2731                 \\
\hline
\end{tabular}
\end{table}


\begin{table}[h!]
\caption{Bias and MSE for the indicated ML estimates of the generalized Heckman-$t$ model parameters (Scenario 3).}
\centering
\footnotesize
\label{tab:mc3}
\begin{tabular}{llccclcccl}
\hline

     &                                  &                 & \multicolumn{2}{c}{Generalized Heckman-$t$} & \multicolumn{1}{c}{}             & \multicolumn{1}{c}{}                & \multicolumn{2}{c}{Generalized Heckman-$t$}          \\
n    & Parameters                       & Censoring & Bias                 & MSE                & Parameters                       & \multicolumn{1}{c}{Censoring} & \multicolumn{1}{c}{Bias} & \multicolumn{1}{c}{MSE} \\
\hline
500  & \multirow{3}{*}{$\beta_1$ \ \ \ \ 1.1}   & 49.8604         & -0.0037              & 0.0094             & \multirow{3}{*}{$\beta_1$ \ \ \ \ 1.1}   & 49.9824                             & -0.0109                  & 0.0072                  \\
1000 &                                  & 49.9469         & -0.0008              & 0.0035             &                                  & 49.9075                             & -0.0049                  & 0.003                   \\
2000 &                                  & 50.0265         & -0.0012              & 0.0017             &                                  & 49.925                              & -0.0046                  & 0.0011                  \\
\hline

500  & \multirow{3}{*}{$\beta_2$ \ \ \ \ 0.7}   & 49.8604         & -0.0025              & 0.0016             & \multirow{3}{*}{$\beta_2$ \ \ \ \ 0.7}   & 49.9824                             & -0.0057                  & 0.0014                  \\
1000 &                                  & 49.9469         & -0.0006              & 0.0005             &                                  & 49.9075                             & -0.0026                  & 0.0005                  \\
2000 &                                  & 50.0265         & -0.0009              & 0.0002             &                                  & 49.925                              & -0.0025                  & 0.0002                  \\
\hline

500  & \multirow{3}{*}{$\beta_3$ \ \ \ \ 0.1}   & 49.8604         & 0.0001               & 0.0014             & \multirow{3}{*}{$\beta_3$ \ \ \ \ 0.1}   & 49.9824                             & 0.0017                   & 0.0011                  \\
1000 &                                  & 49.9469         & <0.0001                    & 0.0005             &                                  & 49.9075                             & 0.0004                   & 0.0004                  \\
2000 &                                  & 50.0265         & 0.0005               & 0.0002             &                                  & 49.925                              & -0.0002                  & 0.0002                  \\
\hline
500  & \multirow{3}{*}{$\gamma_1$ \ \ \ \ 0}    & 49.8604         & 0.0025               & 0.0065             & \multirow{3}{*}{$\gamma_1$ \ \ \ \ 0}    & 49.9824                             & 0.0012                   & 0.0063                  \\
1000 &                                  & 49.9469         & 0.0036               & 0.003              &                                  & 49.9075                             & 0.0029                   & 0.004                   \\
2000 &                                  & 50.0265         & -0.0023              & 0.0015             &                                  & 49.925                              & 0.0022                   & 0.0014                  \\
\hline
500  & \multirow{3}{*}{$\gamma_2$ \ \ \ \ 0.5}  & 49.8604         & 0.0067               & 0.0078             & \multirow{3}{*}{$\gamma_2$ \ \ \ \ 0.5}  & 49.9824                             & 0.0051                   & 0.0071                  \\
1000 &                                  & 49.9469         & 0.0061               & 0.0035             &                                  & 49.9075                             & 0.0036                   & 0.0051                  \\
2000 &                                  & 50.0265         & 0.0009               & 0.0019             &                                  & 49.925                              & 0.0003                   & 0.0015                  \\
\hline
500  & \multirow{3}{*}{$\gamma_3$ \ \ \ \ 1.1}  & 49.8604         & 0.009                & 0.0161             & \multirow{3}{*}{$\gamma_3$ \ \ \ \ 1.1}  & 49.9824                             & 0.0095                   & 0.0155                  \\
1000 &                                  & 49.9469         & 0.0047               & 0.008              &                                  & 49.9075                             & 0.0116                   & 0.0185                  \\
2000 &                                  & 50.0265         & 0.0024               & 0.0037             &                                  & 49.925                              & 0.0046                   & 0.0035                  \\
\hline
500  & \multirow{3}{*}{$\gamma_4$ \ \ \ \ 0.6}  & 49.8604         & 0.0010                & 0.0091             & \multirow{3}{*}{$\gamma_4$ \ \ \ \ 0.6}  & 49.9824                             & 0.0048                   & 0.0083                  \\
1000 &                                  & 49.9469         & 0.0014               & 0.0042             &                                  & 49.9075                             & 0.0056                   & 0.0086                  \\
2000 &                                  & 50.0265         & -0.001               & 0.0023             &                                  & 49.925                              & 0.0038                   & 0.002                   \\
\hline
500  & \multirow{3}{*}{$\kappa_1$ \ \ \ \ -0.3} & 49.8604         & -0.0105              & 0.04               & \multirow{3}{*}{$\kappa_1$ \ \ \ \ -0.7} & 49.9824                             & -0.0166                  & 0.0437                  \\
1000 &                                  & 49.9469         & -0.0107              & 0.0173             &                                  & 49.9075                             & -0.0061                  & 0.0183                  \\
2000 &                                  & 50.0265         & -0.0015              & 0.0085             &                                  & 49.925                              & -0.0034                  & 0.0078                  \\
\hline
500  & \multirow{3}{*}{$\kappa_2$ \ \ \ \ -0.3} & 49.8604         & -0.0302              & 0.0394             & \multirow{3}{*}{$\kappa_2$ \ \ \ \ -0.7} & 49.9824                             & -0.0831                  & 0.0609                  \\
1000 &                                  & 49.9469         & -0.0189              & 0.0163             &                                  & 49.9075                             & -0.0317                  & 0.0195                  \\
2000 &                                  & 50.0265         & -0.0053              & 0.007              &                                  & 49.925                              & -0.0214                  & 0.0086                  \\
\hline
500  & \multirow{3}{*}{$\lambda_1$ \ \ \ \ -0.4} & 49.8604         & 0.0032               & 0.0067             & \multirow{3}{*}{$\lambda_1$ \ \ \ \ -0.4} & 49.9824                             & -0.0041                  & 0.0072                  \\
1000 &                                  & 49.9469         & -0.0003              & 0.0035             &                                  & 49.9075                             & -0.006                   & 0.0034                  \\
2000 &                                  & 50.0265         & 0.0005               & 0.0016             &                                  & 49.925                              & -0.0033                  & 0.0018                  \\
\hline
500  & \multirow{3}{*}{$\lambda_2$ \ \ \ \ 1.2} & 49.8604         & 0.0062               & 0.004              & \multirow{3}{*}{$\lambda_2$ \ \ \ \ 1.2} & 49.9824                             & 0.0039                   & 0.0038                  \\
1000 &                                  & 49.9469         & 0.0057               & 0.0019             &                                  & 49.9075                             & 0.0029                   & 0.0018                  \\
2000 &                                  & 50.0265         & 0.0026               & 0.0009             &                                  & 49.925                              & 0.0016                   & 0.0008                  \\
\hline
500  & \multirow{3}{*}{$\nu$ \ \ \ \ \ 4 }              & 49.8604         & 0.5755               & 3.9653             & \multirow{3}{*}{$\nu$ \ \ \ \ \ 4 }             & 49.9824                             & 0.546                    & 8.083                   \\
1000 &                                  & 49.9469         & 0.2316               & 0.8686             &                                  & 49.9075                             & 0.1677                   & 0.793                   \\
2000 &                                  & 50.0265         & 0.1139               & 0.3197             &                                  & 49.925                              & 0.0841      & 0.4107

\\
\hline
\end{tabular}
\end{table}

\clearpage

\section{Application to real data} \label{sec:5}

In this section, two real data sets, corresponding to outpatient expense and investments in education, are analyzed. The outpatient expense data data set has already been analyzed in the literature by Heckman models \cite{Genton2012}, whereas the education investment data is new and is analyzed for the first time here.

\subsection{Outpatient expense}\label{sec:5.1}

In this subsection, a real data set corresponding to outpatient expense from the 2001 Medical Expenditure Panel Survey (MEPS) is used to illustrate the proposed methodology. This data set has information about the cost and provision of outpatient services, and is the most complete coverage about health insurance in the United States, according to the Agency for Healthcare Research and Quality (AHRQ).

The MEPS data set contains information collected from 3328 individuals between 21 and 64 years. The variable of interest is the expenditure on medical services in the logarithm scale ($Y_{i}^{*}=lnambx$), while the latent variable ($U_{i}^{*} = dambexp$) is the willingness of the individual to spend; $U_i=\mathds{1}_{\{U_{i}^{*}>0\}}$ corresponds the decision of the individual to spend. It was verified that 526 (15.8\%) of the outpatient costs are identified as zero (censored). The covariates considered in the are: $age$ is the age measured in tens of years; $fem$ is a dummy variable that assumed value 1 for women and 0 for men; $educ$ is the years of education; $blhisp$ is a dummy variable for ethnicity (1 for black or Hispanic and 0 if non-black and non-Hispanic); $totcr$ is the total number of chronic diseases; $ins$ is the insurance status; and $revenue$ denotes the individual income.

Table~\ref{tab:estdesc_app1} reports descriptive statistics of the observed medical expenditures, including the minimum, mean, median, maximum, standard deviation (SD), coefficient of variation (CV), coefficient of skewness (CS) and coefficient of (excess) kurtosis (CK) values. From this table, we note the following: the mean is almost equal to the median; a very small negative skewness value; and a very small kurtosis value. The symmetric nature of the data is confirmed by the histogram shown in Figure Figure~\ref{fig:hist_box_app1}(a). The boxplot shown in Figure~\ref{fig:hist_box_app1}(b) indicates some potential outliers.
Therefore, we observe that a symmetric distribution is a reasonable assumption, more specifically a Student-$t$ model, since since we have to accommodate outliers.

\begin{table}[!ht]
\centering
\caption{Summary statistics for the medical expenditure data.}
\label{tab:estdesc_app1}
\begin{tabular}{ccccccccccccccccc}
\hline
 minimum &  Mean   & Median  & maximum  &  SD     &  CV       &  CS    &  CK    &  $n$ \\
0.693    & 6.557     & 6.659 & 10.819   & 1.406   & 21.434\%  & -0.315 & 0.006  &  2801 \\
\hline
\end{tabular}
\end{table}

\begin{figure}[h!]
\vspace{0.25cm}
\subfigure[]{\includegraphics[height=7cm,width=7.25cm]{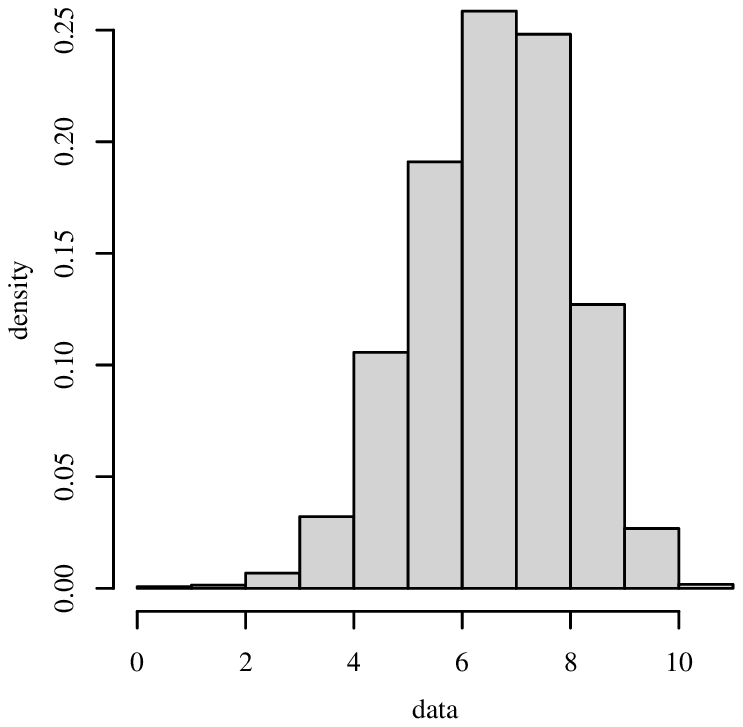}}
\subfigure[]{\includegraphics[height=7cm,width=7.25cm]{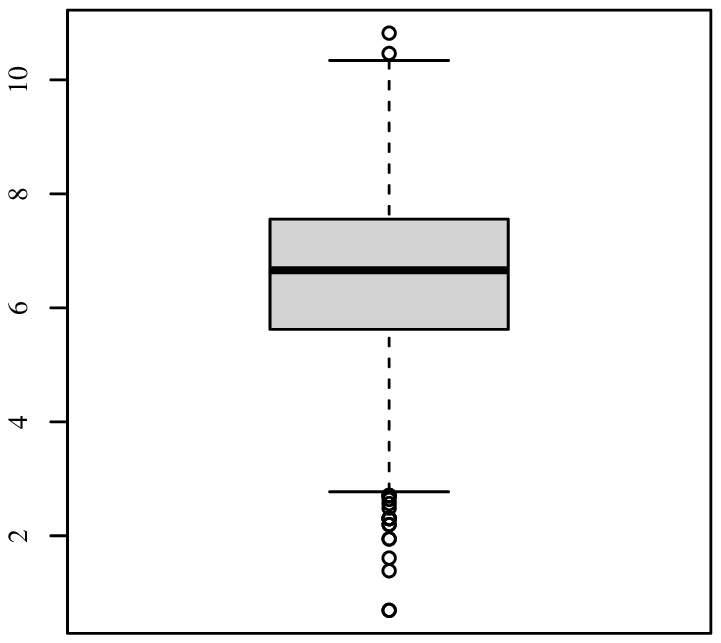}}
\caption{\small Histogram  (a) and boxplots (b) for the medical expenditure data.}
  \label{fig:hist_box_app1}
\end{figure}

We then analyze the medical expenditure data using the generalized Heckman-$t$ model, expressed as
\begin{equation}
\label{eq1}
lnambx=\beta_0+ \beta_1\, age_i+\beta_2\,fem_i+\beta_3\,educ_i+\beta_4\, blhisp_i+\beta_5\,totchr_i+\beta_6\,ins_i,
\end{equation}
\begin{equation}
\label{eq2}
dambexp= \gamma_0+\gamma_1\,age_i+\gamma_2\,fem_i+\gamma_3\,educ_i+\gamma_4\,blhisp_i+\gamma_5\,totchr_i+\gamma_6\,ins_i+\gamma_7\,income_i,
\end{equation}
\begin{equation}
    \log \sigma_i = \lambda_0 + \lambda_1\, age_i + \lambda_2\, totchr_i + \lambda_3\, ins_i,
\end{equation}
\begin{equation}
\text{arctanh}\, \rho_i = \kappa_0 + \kappa_1\, fem_i + \kappa_2\, totchr_i.
\end{equation}

We initially compare the adjustments of the generalized Heckman-$t$ (GH$t$) model, in terms of Akaike (AIC) and Bayesian information (BIC), with the adjustments of the classical Heckman-normal (CHN) \citep{heckman76} and generalized Heckman-normal (GHN) \citep{Bastos2021} models; see Table \ref{AIC_BIC_meps}. The AIC and BIC values reveal that the GH$t$ model provides the best adjustment, followed by the GHN model.

\def\tablename{Table}
\begin{table}[h!]
\caption{AIC and BIC of the indicated Heckman models.}
\centering
\small
\begin{tabular}{ccccccccc}
\hline

    & CHN & GHN      & GH$t$    \\
\hline
AIC & 11706.44   & 11660.29 & 11635.05\\
BIC & 11810.31   & 11794.71  & 11775.59 \\
\hline
\end{tabular}
\label{AIC_BIC_meps}
\end{table}

Table \ref{tab_meps} presents the estimation results of the GHN and GH$t$ models. From this table, we observe the following results: the explanatory variables $totchr$ and $ins$ that model the dispersion are significant, in both models, indicating the presence of heteroscedasticity in the data. The explanatory variable $age$ is only significant in the GHN model. For the correlation term, the covariates $fem$ and $totchr$ are significant for both models, which indicates the presence of selection bias in the data.

In the outcome equation, when we look at both models, $age$, $fem$, $blhisp$ and $totchr$ are significant at the 5\% level, and $educ$ is not significant. The explanatory variable $ins$ is significant at the 5\% and 10\% levels in the GHN and GH$t$ models, respectively; see Table \ref{tab_meps}. We can interpret the estimated coefficients in terms of the effect on the expenditure on medical services; see \citet[p.82]{weisberg:14}. For example, a 1-year increase in $age$ rises by
$(\exp(0.1838)-1)\times100 = 20.18\%$ and $(\exp(0.1895)-1)\times100 = 20.86\%$ the expected value of the expenditure on medical services according to the GHN and GH$t$, respectively. Moreover, a 1-unit increase in $totchr$ rises by
$(\exp(0.4306)-1)\times100 = 53.82\%$ and $(\exp(0.4464)-1)\times100 = 56.27\%$ the expected value of the response according to the GHN and GH$t$ models, respectively.

In the case of the selection equation, we observe that, for both models, the explanatory variables $age$, $fem$, $educ$, $blhisp$, $totchr$ and $ins$ are significant at the 5\% level and $revenue$ is significant at the 10\% level; see Table \ref{tab_meps}. The interpretation is made in terms of odds ratio, which is obtained by exponentiating the estimated explanatory variable coefficient. For example, for the GH$t$ case, the odds ratio for $age$ is $\exp(0.0930) = 1.0975$, suggesting that each additional year of age raises the likelihood of an individual having expenditures on medical services by $ ((1.0975 - 1)\times  100)=9.75\%$.

\def\tablename{Table}
\begin{table}[h!]
\caption{Estimation results of the GHN and GH$t$ models.}
\centering
\scriptsize
\begin{tabular}{lrrrrrrrrrrrrr}

\hline
\\
\vspace{-.5cm}
\\
\multicolumn{9}{c}{\small{Probit selection equation}}
\vspace{0.1cm}
\\
\hline

\multirow{2}{*}{Variables} & \multicolumn{2}{c}{Estimates} & \multicolumn{2}{c}{Std. Error} & \multicolumn{2}{c}{$t$ Value} & \multicolumn{2}{c}{$p$-value}               \\
                           & GHN   & GH$t$    & GHN    & GH$t$    & GHN   & GH$t$   & GHN         & GH$t$          \\
\hline
$(Intercept)$ & -0.5903 & -0.6406 & 0.1867 & 0.2014 & -3.1620 & -3.1800 & \textless{}0.001            & \textless{}0.001           \\
$age$         & 0.0863  & 0.0930  & 0.0264 & 0.0288 & 3.2610  & 3.2230  & \textless{}0.001            & \textless{}0.001           \\
$fem$         & 0.6299  & 0.7087  & 0.0597 & 0.0681 & 1.0543 & 1.0395 & \textless{}0.001 & \textless{}0.001  \\
$educ$        & 0.0569 & 0.0590  & 0.0114 & 0.0122 & 4.9830  & 4.8110  & \textless{}0.001         & \textless{}0.001        \\
$blhisp$      & -0.3368 & -0.3726 & 0.0596 & 0.0647 & -5.6430 & -5.7590 & \textless{}0.001        & \textless{}0.001          \\
$totchr$      & 0.7585  & 0.8728  & 0.0686 & 0.0858 & 1.1043 & 1.0168 & \textless{}0.001  &\textless{}0.001  \\
$ins$         & 0.1727  & 0.1863  & 0.0611 & 0.0665 & 2.8240  & 2.7990  & \textless{}0.001            & \textless{}0.001           \\
$revenue$      & 0.0022  & 0.0025 & 0.0012 & 0.0013 & 1.8380  & 1.8750  & 0.0661            & 0.0610     \\
\hline
\\
\vspace{-.5cm}
\\
\multicolumn{9}{c}{\small{Outcome equation}}
\vspace{0.1cm}
\\
\hline
\multirow{2}{*}{Variables} & \multicolumn{2}{c}{Estimates} & \multicolumn{2}{c}{Std. Error} & \multicolumn{2}{c}{$t$ Value} & \multicolumn{2}{c}{$p$-value}               \\

                           & GHN    & GH$t$    & GHN    & GH$t$   & GHN    & GH$t$   & GHN         & GH$t$          \\
\hline
$(Intercept)$& 5.7041   & 5.6078   & 0.1930 & 0.1912 & 2.9553 & 2.9316 & \textless{}0.001  & \textless{}0.001 \\
$age$       & 0.1838  & 0.1895  & 0.0234 & 0.0230 & 7.8460  & 8.2350  & \textless{}0.001          & \textless{}0.001         \\
$fem$    & 0.2498  & 0.2555  & 0.0587 & 0.0580 & 4.2530  & 4.4000  & \textless{}0.001          & \textless{}0.001          \\
$educ$      & 0.0013  & 0.0062  & 0.0101 & 0.0100 & 0.1290 & 0.6240 & 0.8970              & 0.5329              \\
$blhisp$    & -0.1283 & -0.1344 & 0.0577 & 0.0569 & -2.2210 & -2.3630 & 0.0264             & 0.0182             \\
$totchr$    & 0.4306  & 0.4464  & 0.0305 & 0.0297 & 1.4115 & 1.5002 & \textless{}0.001  & \textless{}0.001  \\
$ins$       & -0.1027 & -0.0976 & 0.0513 & 0.0501 & -2.000 & -1.9470 & 0.0456             & 0.0516    \\
\hline
\\
\vspace{-.5cm}
\\
\multicolumn{9}{c}{\small{Dispersion}}
\vspace{0.1cm}
\\
\hline
\multirow{2}{*}{Variables} & \multicolumn{2}{c}{Estimates} & \multicolumn{2}{c}{Std. Error} & \multicolumn{2}{c}{$t$ Value} & \multicolumn{2}{c}{$p$-value}
\\

                           & GHN     & GH$t$    & GHN     & GH$t$    & GHN   & GH$t$   & GHN   & GH$t$   \\
\hline

$Intercept)$ & 0.5081  & 0.4172  & 0.0573 & 0.0643 & 8.8550  & 6.4820  & \textless{}0.001  & \textless{}0.001  \\
$age$        & -0.0249 & -0.0209 & 0.0125 & 0.0136 & -1.9870 & -1.5360 & 0.0469           & 0.1246    \\
$totchr$     & -0.1046 & -0.1118 & 0.0191 & 0.0208 & -5.4760 & -5.3720 & \textless{}0.001         & \textless{}0.001  \\
$ins$        & -0.1070 & -0.1117 & 0.0277 & 0.0303 & -3.8630 & -3.6810 & \textless{}0.001         & \textless{}0.001  \\
\hline
\\
\vspace{-.5cm}
\\
\multicolumn{9}{c}{\small{Correlation}}
\vspace{0.1cm}
\\
\hline
\multirow{2}{*}{Variables} & \multicolumn{2}{c}{Estimates} & \multicolumn{2}{c}{Std. Error} & \multicolumn{2}{c}{$t$ Value} & \multicolumn{2}{c}{$p$-value} \\
                           & GHN     & GH$t$    & GHN     & GH$t$    & GHN   & GH$t$   & GHN   & GH$t$   \\
\hline
$Intercept$  & -0.6475 & -0.6051 & 0.1143 & 0.1118 & -5.666 & -5.413 & \textless{}0.001  & \textless{}0.001  \\
$fem$     & -0.4040 & -0.4220 & 0.1356 & 0.1489 & -2.978 & -2.835 & \textless{}0.001   & \textless{}0.001   \\
$totchr$     & -0.4379 & -0.4999 & 0.1862 & 0.2102 & -2.351 & -2.378 & 0.0187    & 0.0174   \\
\hline
$\nu$ & \ \ \ \ - &  12.3230   & \ \ \ \ - &  2.7570 & \ \ \ \ - & 4.4690 & \ \ \ \ - & \textless{}0.001 \\

\hline
\end{tabular}
\label{tab_meps}
\end{table}

Figure~\ref{fig:residuals1} displays the quantile versus quantile (QQ) plots of the martingale-type (MT) residuals for the GHN and GH$t$ models. This residual is given by

\begin{equation}
\label{res}
{r}_{{i}}^{\text{MT}}=\textrm{sign}(r^{_{\textrm{\tiny
M}_{i}}})\sqrt{-2\left( r^{_{\textrm{\tiny
M}_{i}}}+u_i\log(u_i- r^{_{\textrm{\tiny M}_{i}}})\right)},\quad i=1,\dots,n.
\end{equation}
where $r^{_{\textrm{\tiny M}_{i}}} = u_i + \log(\widehat S(t_i))$, $\widehat S(t_i)$ is the fitted survival function, and $u_i=0$ or $1$ indicating that case $i$ is censored or not, respectively; see \cite{tgf:90}. The MT residual is asymptotically standard normal, if the model is correctly specified whatever the specification of the model is. From Figure~\ref{fig:residuals1}, we see clearly that the GH$t$ model provides better fit than GHN model.

\begin{figure}[h!]
    \centering
    \subfigure[GHN]{\includegraphics[width=7cm, height=7cm]{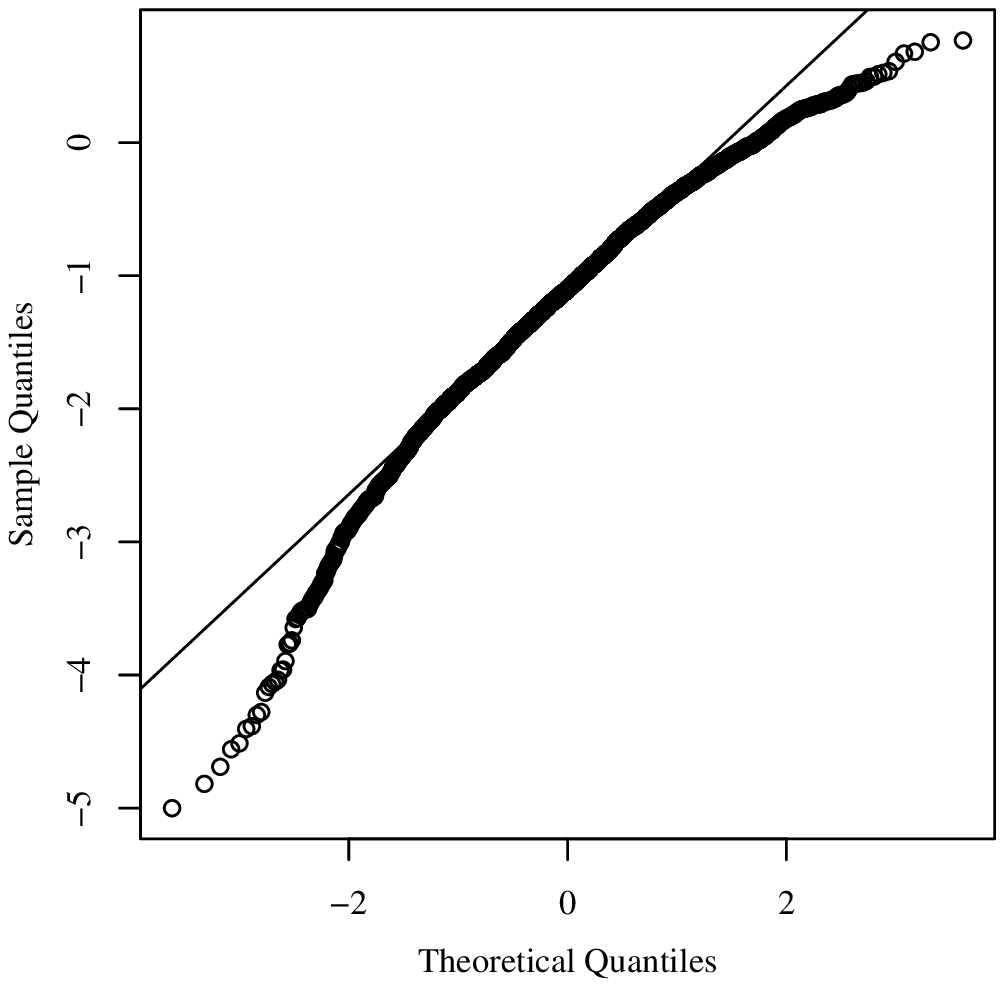}}
    \subfigure[GH$t$]{\includegraphics[width=7cm, height=7cm]{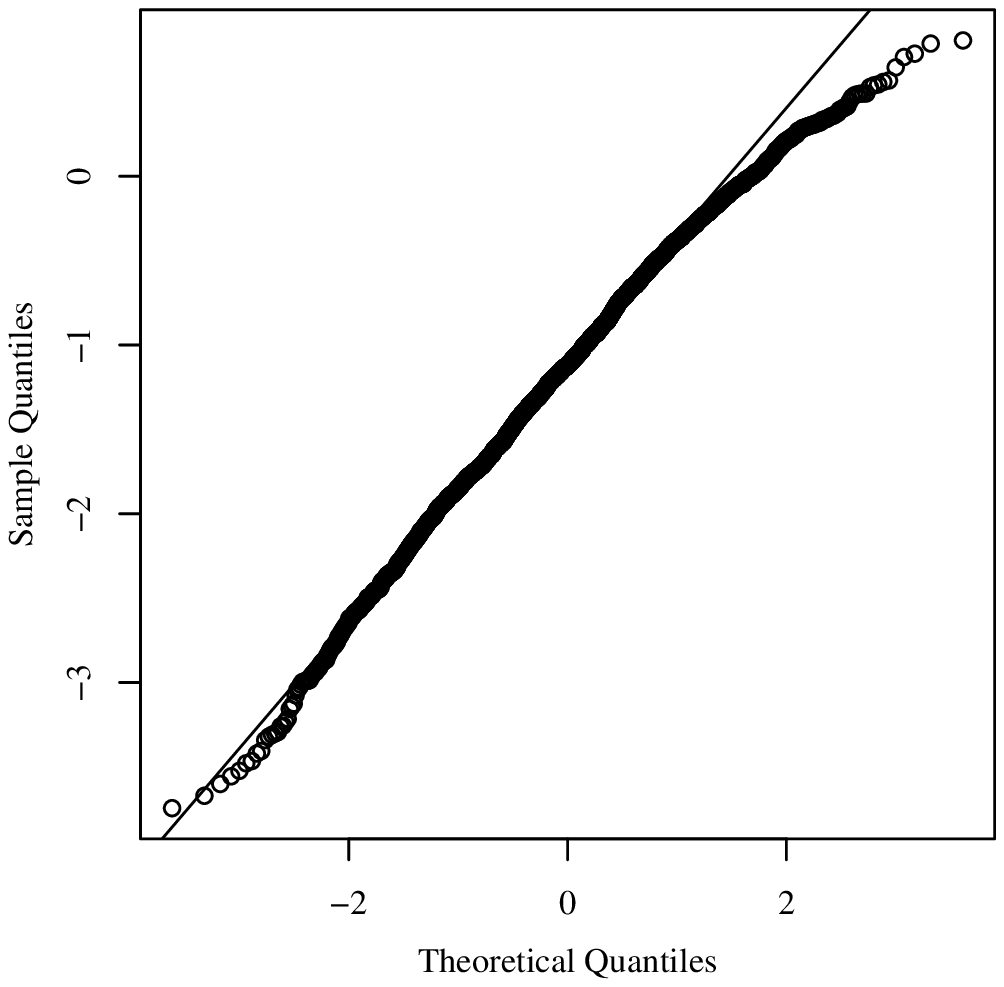}}
    \caption{QQ plot for the MT residuals for the GHN and GH$t$ models.}
    \label{fig:residuals1}
\end{figure}


\subsection{Investments in education}\label{sec:5.1}

\noindent
In this subsection, data on education investments are used to illustrate the proposed methodology. We consider the investments made by municipalities of two Brazilian states: Sao Paulo (SP) and Minas Gerais (MG). This data set was obtained from the Brazilian National Fund for Educational Development (FNDE) website\footnotemark[1], which is a federal agency under the Ministry of Education. The origin of these investments comes from the Fund for the Maintenance and Development of Basic Education and the Valorization of Education Professionals (FUNDEB). The resources are distributed to 27 federative units (26 states plus the Federal District), according to the number of students enrolled in their basic education network. This rule is established for the previous year's school census data, e.g., 2018 resources were based on 2017 student numbers. This method helps to better distribute resources across the country, as it takes into account the size of education networks.

The variable of interest is education investments with 1503 observations, of which 102 (7\%) correspond to unobserved investment values identified as zero investment. The explanatory variables considered in the study were: $revenue$\footnotemark[3] represents per capita revenue collected by the municipality; $gnp$ \footnotemark[3] is the Gross National Product of the municipality; $distribute$\footnotemark[2] is a dummy variable indicating if the municipality receives resources from the Financial Compensation for Exploration of Mineral Resources (CFEM); this resource must be destined to investments in the areas of health, education and infrastructure for the community; $sp$ is an indicator variable for state ($sp$ receives value 1); $enrollment$\footnotemark[4] is the school census enrollment numbers.

As in the previous study, the response variable, investment in education, is in the logarithm scale $Y_{i}^{*}=lninvest$. The latent variable ($U_{i}^{*} = dinvest$) denotes the willingness of the $i$th municipality to invest education; $U_i=\mathds{1}_{\{U_{i}^{*}>0\}}$ corresponds to the decision or not of the $i$th municipality to invest in education.

\footnotetext[1]{Filtered data are available at \hbox{https://repositorio.shinyapps.io/plataforma\_de\_dados\_municipais}.}

\footnotetext[2]{\hbox{https://dados.gov.br/dataset/sistema-arrecadacao}.}

\footnotetext[3]{\hbox{http://www.ipeadata.gov.br/Default.aspx}.}

\footnotetext[3]{\hbox{https://www.ibge.gov.br/estatisticas/sociais/populacao/9103-estimativas-de-populacao.html?=\&t=resultados}.}

\footnotetext[4]{\hbox{https://www.gov.br/inep/pt-br/areas-de-atuacao/pesquisas-estatisticas-e-indicadores/censo-escolar}.}

The descriptive statistics for the investments in education are reported in Table~\ref{tab:estdesc_app1_ie}. From this table, we note the following: the mean is almost equal to the median; a very small negative skewness value, and a high kurtosis value. The symmetric nature of the data is confirmed by the histogram shown in Figure Figure~\ref{fig:hist_box_app1_ie}(a). The boxplot shown in Figure~\ref{fig:hist_box_app1_ie}(b) indicates potential outliers. Therefore, we observe that a Student-$t$ model is a reasonable assumption.

\begin{table}[!ht]
\centering
\caption{Summary statistics for the education investment data.}
\label{tab:estdesc_app1_ie}
\begin{tabular}{ccccccccccccccccc}
\hline
 minimum &  Mean   & Median  & maximum  &  SD     &  CV       &  CS    &  CK    &  $n$ \\
4.271    & 12.212    & 12.575 & 18.511   & 1.903   &  15.587\%  & -0.5864 &  3.70 &  1401 \\
\hline
\end{tabular}
\end{table}

\begin{figure}[h!]
\vspace{0.25cm}
\subfigure[]{\includegraphics[height=7cm,width=7.25cm]{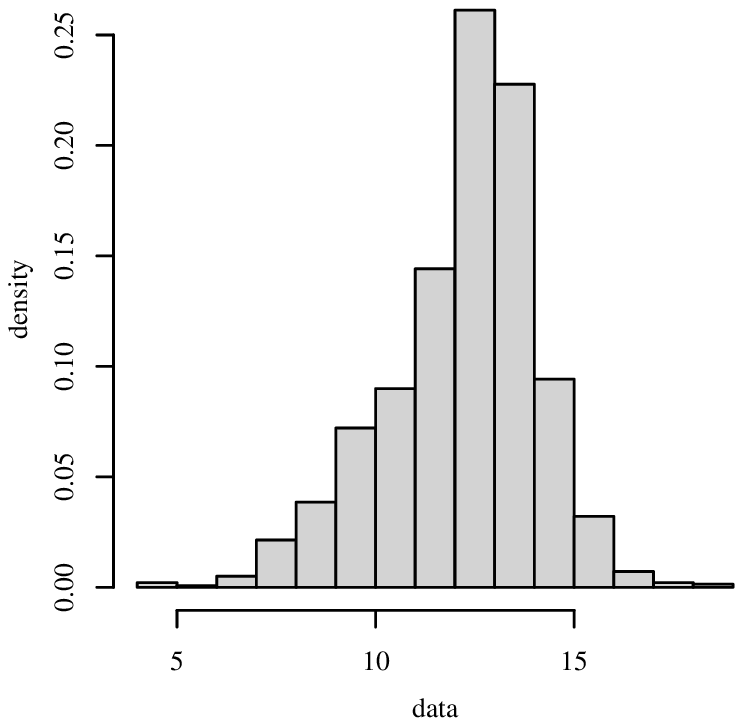}}
\subfigure[]{\includegraphics[height=7cm,width=7.25cm]{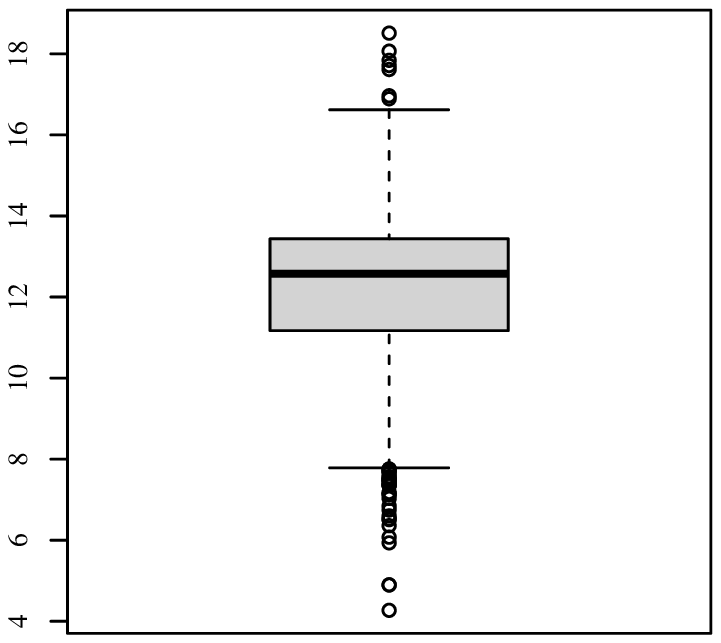}}
\caption{\small Histogram  (a) and boxplots (b) for the education investments data.}
  \label{fig:hist_box_app1_ie}
\end{figure}

We then analyze the education investment data using the GH$t$ model, expressed as
\begin{equation}
\label{eq1}
dinvest = \beta_0+ \beta_1\, revenue_i + \beta_2\, sp_i  + \beta_3\, enrollment_i,
\end{equation}
\begin{equation}
\label{eq2}
lninvest = \gamma_0+\gamma_1\, revenue_i + \gamma_2\, sp_i+\gamma_3\, enrollment_i+ \gamma_4\, gnp_i,
\end{equation}
\begin{equation}
    \log \sigma_i = \lambda_0 + \lambda_1\, revenue_i + \lambda_2\, sp_i,
\end{equation}
\begin{equation}
\text{arctanh}\, \rho_i = \kappa_0 + \kappa_1 \, revenue_i + \kappa_2\, distribute_i.
\end{equation}

Table \ref{AIC_BIC_IE} reports the AIC and BIC of the CHN, GHN and GH$t$ models. From Table
\ref{AIC_BIC_IE}, we observe that the GH$t$ model provides better adjustment than other
models based on the values of AIC and BIC.

\def\tablename{Table}
\begin{table}[h!]
\caption{AIC and BIC of the indicated Heckman models.}
\centering
\small
\begin{tabular}{lllll}
\hline

    & CHN & GHN      & GH$t$    \\
\hline
AIC &  8269.058   & 5873.579 & 5719.590\\
BIC &  8306.147   & 5953.055 & 5804.365 \\
\hline
\end{tabular}
\label{AIC_BIC_IE}
\end{table}

Table \ref{tab_invest} presents the estimation results of the GHN and GH$t$ models. From this table, we note that the explanatory variables $revenue$ and $sp$ associated with the dispersion are significant, in both models, suggesting the presence of heteroscedasticity. For the explanatory variables related to the correlation parameter, $revenue$ is significant at the 5\% and 10\% levels in the GHN and GH$t$ models, respectively, while $distribute$ is significant at the 10\% and 5\% levels in the GHN and GH$t$ models, respectively. Therefore, there is evidence of the presence of
sample selection bias in the data.

From Table \ref{tab_invest}, we observe that in the outcome equation $revenue$, $sp$, $enrollment$ and $gnp$ are significant, according to the GHN and GH$t$ models. Note that increasing $revenue$ by one unit is associated with
$(\exp(-0.2683)-1)\times100 = -23.53\%$ and $(\exp(-0.2167)-1)\times100 = -19.48\%$ decrease in the average investment in education according to the GHN and GH$t$ models, respectively. Note also that when the municipality is located in the state of Sao Paulo ($sp$), the average investment in education increases by $(\exp(1.0714)-1)\times100 = 191.95\%$ and $(\exp(1.0063)-1)\times100 = 173.55\%$, according to the GHN and GH$t$ models, respectively. These results show that the municipalities of the state of Sao Paulo invest much more in education than the municipalities of the state of Minas Gerais.

In the selection equation, we observe that the explanatory variables $revenue$, $sp$ and $enrollment$ are significant in the GH$t$ model. On the other hand, $revenue$ and $sp$ are not significant in the GHN model; see Table \ref{tab_invest}. We observe that, for the GH$t$ case, the odds ratio for $sp$ is $\exp(0.3674) = 1.4439$, that is, the likelihood of a municipality to spend on medical services increases by $((1.4439 -1) \times 100)=44.39\%$ when the municipality is located in the state of Sao Paulo

\def\tablename{Table}
\begin{table}[h!]
\caption{Estimation results of the GHN and GH$t$ models.}
\centering
\scriptsize
\begin{tabular}{lrrrrrrrrrrrrrrrrrr}
\hline
\\
\vspace{-.5cm}
\\
\multicolumn{9}{c}{\small{Probit selection equation}}
\vspace{0.1cm}
\\
\hline

\multirow{2}{*}{Variables} & \multicolumn{2}{c}{Estimate} & \multicolumn{2}{c}{Std. Error} & \multicolumn{2}{c}{t Value} & \multicolumn{2}{c}{p-Value}       \\
                           & GHN           & GH$t$        & GHN            & GH$t$         & GHN          & GH$t$        & GHN             & GH$t$           \\
\hline
$(intercept)$              & 1.0523      & 1.4222    & 0.1013       & 0.0955      & 10.3870       & 14.8880       & \textless{}0.01 & \textless{}0.001 \\
$revenue$                   & -0.0058     & -0.0530    & 0.0223       & 0.0164     & -0.2630       & -3.2330       & 0.7920           & \textless{}0.001 \\
$sp$                       & 0.1065      & 0.3674     & 0.0990       & 0.1032      & 1.0750        & 3.5600         & 0.2820           & \textless{}0.001 \\
$enrollment$               & 0.0311      & 0.0537     & 0.0032      & 0.0026      & 9.6380        & 20.6150       & \textless{}0.001 & \textless{}0.001 \\

\hline
\\
\vspace{-.5cm}
\\
\multicolumn{9}{c}{\small{Outcome equation}}
\vspace{0.1cm}
\\
\hline
\multirow{2}{*}{Variables} & \multicolumn{2}{c}{Estimate} & \multicolumn{2}{c}{Std. Error} & \multicolumn{2}{c}{t Value} & \multicolumn{2}{c}{p-Value}       \\

                           & GHN           & GH$t$        & GHN            & GH$t$         & GHN          & GH$t$        & GHN             & GH$t$           \\
\hline
$(intercept)$              & 12.4789    & 12.4058   & 0.1252      & 0.1208    & 99.6480       & 102.6960      & \textless{}0.001 & \textless{}0.001 \\
$revenue$                   & -0.2683    & -0.2167  & 0.0332      & 0.0274     & -8.0640       & -7.9030       & \textless{}0.001 & \textless{}0.001 \\
$sp$                       & 1.0714     & 1.0063   & 0.1003      & 0.0881     & 10.6810       & 11.4190       & \textless{}0.001 & \textless{}0.001 \\
$enrollment$               & 0.0031     & 0.0226    & 0.0005      & 0.0025     & 5.6670        & 8.8410         & \textless{}0.001 & \textless{}0.001 \\
$gnp$                      & 0.0202     & 0.0151    & 0.0023      & 0.0005     & 8.4630        & 27.4250       & \textless{}0.001 & \textless{}0.001 \\
\hline
\\
\vspace{-.5cm}
\\
\multicolumn{9}{c}{\small{Dispersion}}
\vspace{0.1cm}
\\
\hline
\multirow{2}{*}{Variables} & \multicolumn{2}{c}{Estimate} & \multicolumn{2}{c}{Std. Error} & \multicolumn{2}{c}{t Value} & \multicolumn{2}{c}{p-Value}       \\

                           & GHN           & GH$t$        & GHN            & GH$t$         & GHN          & GH$t$        & GHN             & GH$t$           \\
\hline
$(intercept)$              & 0.5921       & 0.2947      & 0.0503       & 0.0623       & 11.7560       & 4.7250        & \textless{}0.001 & \textless{}0.001 \\
$revenue$                   & 0.0382       & 0.0546      & 0.0130        & 0.0146       & 2.9330         & 3.7180        & \textless{}0.001 & \textless{}0.001 \\
$sp$                       & -0.3142      & -0.4745     & 0.0416        & 0.0534       & -7.5470       & -8.8800       & \textless{}0.001 & \textless{}0.001 \\
\hline
\\
\vspace{-.5cm}
\\
\multicolumn{9}{c}{\small{Correlation}}
\vspace{0.1cm}
\\
\hline
\multirow{2}{*}{Variables} & \multicolumn{2}{c}{Estimate} & \multicolumn{2}{c}{Std. Error} & \multicolumn{2}{c}{t Value} & \multicolumn{2}{c}{p-Value}       \\
                           & GHN           & GH$t$        & GHN            & GH$t$         & GHN          & GH$t$        & GHN             & GH$t$           \\
\hline
$(intercept)$              & -3.7263      & -6.7669      & 0.5463        & 0.9682        & -6.8210       & -6.9890       & \textless{}0.001 & \textless{}0.001 \\
$revenue$                   & 0.1686       & 0.2274       & 0.0839        & 0.1161        & 2.0080        & -6.9890       & 0.0448          & 0.0500         \\
$distribute$               & 0.8298       & 3.1118       & 0.4411        & 0.6867        & 1.8810        & 1.9580        & 0.0602         & \textless{}0.001 \\
\hline
$\nu$                         &    \ \ \ \    -        & 3.6240        &     \ \ \ \    -        & 0.4320         &  \ \ \ \    -         & 8.3900         &     \ \ \ \   -          & \textless{}0.001
\\
\hline
\end{tabular}
\label{tab_invest}
\end{table}

Figure \ref{fig:residuals2} displays the QQ plots of the MT residuals. This figure indicates that the MT residuals in the GH$t$ model shows better agreement with the reference distribution.

\begin{figure}[h!]
    \centering
    \subfigure[GHN]{\includegraphics[width=7cm, height=7cm]{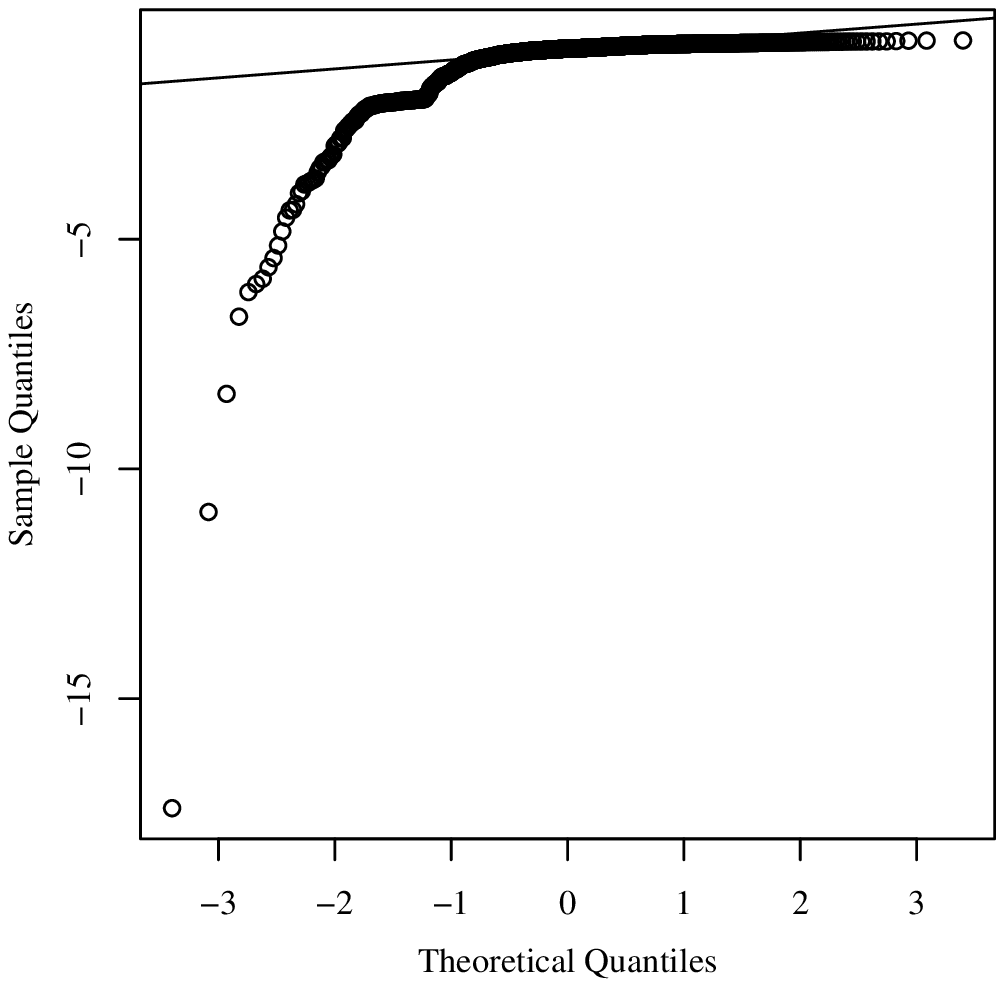}}
    \subfigure[GH$t$]{\includegraphics[width=7cm, height=7cm]{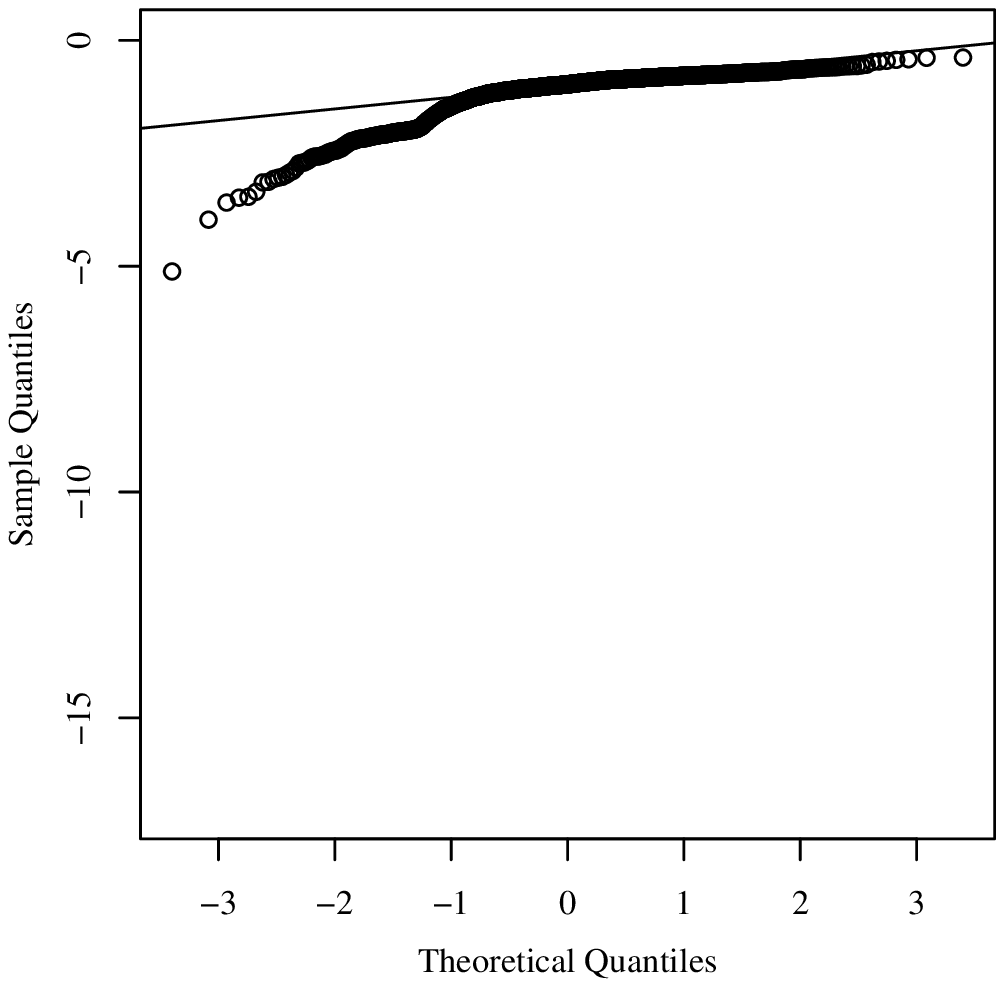}}
    \caption{QQ plot for the MT residuals for the GHN and GH$t$ models.}
    \label{fig:residuals2}
\end{figure}


\section{Concluding Remarks} \label{sec:6}
\noindent

In this paper, a class of Heckman sample selection models were proposed based symmetric distributions. In such models, covariates were added to the dispersion and correlation parameters, allowing the accommodation of heteroscedasticity and a varying sample selection bias, respectively. A Monte Carlo simulation study has showed good results of the parameter estimation method. We have considered high/low censoring rates and the presence of strong/weak correlation. We have applied the proposed model along with some other two existing models to two data sets corresponding to outpatient expense and investments in education. The applications favored the use of the proposed generalized Heckman-$t$ model over the classical Heckman-normal and generalized Heckman-normal models. As part of future research, it will be of interest to propose sample selection models based on skew-symmetric distributions. Furthermore, the behavior of the Wald, score, likelihood ratio and gradient tests can be investigated. Work on these problems is currently in progress and we hope to report these findings in future.

\normalsize


\end{document}